\documentclass[10pt,journal]{IEEEtran}
\makeatletter
\usepackage{epsfig}
\usepackage{epstopdf}
\usepackage{times}
\usepackage{algorithm}
\usepackage{algorithmic}
\usepackage{amssymb}
\usepackage{amsmath}
\usepackage{subfigure}
\usepackage{multirow}
\usepackage{url}
\usepackage{eso-pic}
\usepackage{lipsum}
\usepackage{soul}
\usepackage{caption}

\usepackage{array} 
\def\ignore#1\endignore{}
\newcolumntype{h}{@{}>{\ignore}l<{\endignore}} 

\newcolumntype{x}[1]{%
{\centering\hspace{0pt}}p{#1}}%

\usepackage{latexsym}
\usepackage{multicol}
\usepackage{eurosym}
\usepackage{amsthm}
\usepackage{xspace}
\usepackage{pdfpages}
\usepackage{verbatim}
\usepackage{slashbox}

\newenvironment{sketch}{{\noindent \it Sketch of Proof:~}}

\newtheorem*{rep@theorem}{\rep@title}
\newcommand{\newreptheorem}[2]{%
\newenvironment{rep#1}[1]{%
 \def\rep@title{#2 \ref{##1}}%
 \begin{rep@theorem}}%
 {\end{rep@theorem}}}

\newtheorem{theorem}{Theorem}
\newtheorem{definition}{Definition}
\newtheorem{lemma}{Lemma}
\newreptheorem{lemma}{Lemma}
\renewcommand{\baselinestretch}{0.98}

\begin{document}
\title{A Multi-traffic Inter-cell Interference \\ Coordination Scheme in Dense Cellular Networks}
\author{Vincenzo~Sciancalepore,~\IEEEmembership{Member,~IEEE,}
        Ilario~Filippini,~\IEEEmembership{Senior Member,~IEEE,}
        Vincenzo~Mancuso,~\IEEEmembership{Member,~IEEE,}
        Antonio~Capone,~\IEEEmembership{Senior Member,~IEEE} Albert~Banchs,~\IEEEmembership{Senior Member,~IEEE}
	\thanks{V.~Sciancalepore is with NEC Europe Ltd. (vincenzo.sciancalepore@neclab.eu).
	
V.~Mancuso is with IMDEA Networks Institute, Spain.

A.~Banchs is with the University Carlos III of Madrid and with IMDEA Networks Institute, Spain. 

I.~Filippini and A. Capone are with Politecnico di Milano, Italy. 
}
}
\maketitle
\psfull

\setlength{\textfloatsep}{5pt}

\begin{abstract}
This paper proposes a novel semi-distributed and practical ICIC scheme based on the Almost Blank SubFrame (ABSF) approach specified by 3GPP. We define two mathematical programming problems for the cases of guaranteed and best-effort traffic, and use game theory to study the properties of the derived ICIC distributed schemes, which are compared in detail against unaffordable centralized schemes. 
Based on the analysis of the proposed models, we define {\it Distributed Multi-traffic Scheduling} (DMS), a unified distributed framework for adaptive interference-aware scheduling of base stations in future cellular networks which accounts for both guaranteed and best-effort traffic. 
DMS follows a two-tier approach, consisting of local ABSF \emph{schedulers}, which perform the resource distribution between guaranteed and best effort traffic, and a lightweight local \emph{supervisor}, which coordinates ABSF local decisions. 
As a result of such a two-tier design, DMS requires very light signaling to drive the local schedulers to globally efficient operating points. 
As shown by means of numerical results, DMS allows to $(i)$ maximize radio resources reuse, $(ii)$ provide requested quality for guaranteed traffic, $(iii)$ minimize the time dedicated to guaranteed traffic to leave room for best-effort traffic, and $(iv)$ maximize resource utilization efficiency for best-effort traffic.
\end{abstract}
\thispagestyle{empty}	

\section{Introduction}
\label{s:intro}
The very fast growth of mobile data traffic and the increasing expectations of end users for high rates are pushing wireless industry to speed-up the introduction of new cellular technologies. 
Indeed, it is commonly recognized that the new challenges posed by mobile traffic can be handled only with new technologies and network architectures~\cite{cisco}. This is driving the evolution towards two main directions: on one side the use of high-frequency spectrum portions even under harsh signal propagation~\cite{5G-MiWEBA}, and on the other side the densification of network deployments with a very large number of base stations~\cite{LTEadv}. The latter scenario clearly exacerbates interference issues in traditional mobile spectrum portions and calls for novel resource management schemes. 
To address this point, we focus on interference control for future cellular networks.

As basic approach, 3GPP has proposed the {\it Almost Blank Sub-Frame} (ABSF~\cite{absf}) scheme to easily allow coexistence between macro-cells and limited-power cells (e.g., small or femto-cells) in heterogeneous networks. In particular, the ABSF is used to prevent base station transmissions in selected time slots, thereby reducing the inter-cell interference towards small-cell users. Recently, ABSF has been extensively adopted as a novel inter-cell interference coordination scheme even for homogeneous networks, wherein only macro-cells are deployed (\cite{GMO_2016,SciancaleporeMBZCP14,SMB13}). This novel ABSF standpoint poses new challenges while designing efficient mechanisms for automatically selecting ABSF transmission patterns.

ABSF schemes that have been proposed so far mainly approach the problem from a centralized point of view. This requires a huge exchange of  Channel State Information (CSI) messages 
~\cite{xueying2011cell,irmer2011coordinated,ic_ofdma} and poses scalability issues in dense deployments~\cite{singh2014joint,yinghao2013joint}. Actually, ABSF scheduling is known to be a NP-hard problem~\cite{SciancaleporeMBZCP14} and tackling it for large portions or the entire network is unpractical. 
Moreover, existing schemes compute ABSF patterns without explicitly considering quality guarantees, but they rather aim to increase the spectral efficiency through stochastic approaches, assuming worst case interference conditions~\cite{SMB13}. 
By including traffic guarantees in the loop, the complexity of a centralized approach would increase even further. In contrast, a distributed approach with local decisions would not only be aligned with the well accepted self-organizing network concepts~\cite{son_commag}, but it would also allow to jointly decide in real-time ABSF patterns and user scheduling, rather than assuming worst case conditions for the scheduling process. 
However, a critical aspect in the design of a distributed scheme is the amount of signaling information between base stations. 

In this paper, we propose Distributed Multi-traffic Scheduling (DMS), a resource management scheme providing a lightweight ABSF coordination of \emph{local schedulers} (base stations) with the help of a \emph{supervisor}, which guides ABSF decisions of the base stations and drives the system to a high performance operating point while avoiding fully centralized decisions on ABSF patterns. Hence, DMS defines a semi-distributed approach that offloads and reduces the computational burden from a centralized controller to base stations, while drastically abating the signaling overhead. This makes our approach the first proposal towards a practical and effective solution to ABSF that can be implemented in real networks.

The design of DMS is driven by a game theory-based analysis of two optimization problems: ($i$) minimizing the use of resources for guaranteed traffic and ($ii$) maximizing the efficiency of best-effort communications in the resources left available by the guaranteed traffic. In our model, each base station is a player whose ``moves'' consist in selecting ABSF transmission patterns and announcing them to the neighbors. The supervisor simply instructs the base stations on the amounts of resources to be dedicated to guaranteed traffic and to best-effort traffic. In particular, the base stations play two games: first they play a {\it Distributed Inelastic Game} $\Gamma$ to make decisions on guaranteed traffic allocations, and then they play an {\it Interference Coordination Game} $\Omega$ to decide how to allocate best effort traffic across base stations on the remaining resources. 
By using a specific class of best response strategies, we show that $\Gamma$ always converges to a Nash equilibrium with an algorithm that compacts TTIs used by ABSF patterns of all base stations and leaves as much room as possible for best-effort traffic. We also show that $\Omega$ is a weighted player-specific bottleneck matroid congestion game, which requires the presence of a high-level supervisor to converge.

We validate the proposed scheme via simulation and show that, despite its low complexity and the very limited amount of control messages required, DMS achieves near-optimal performance in terms of:  
$(i)$ maximizing radio resources reuse, $(ii)$ providing required quality to guaranteed traffic, $(iii)$ minimizing the time dedicated to guaranteed traffic to leave room for best-effort traffic, and $(iv)$ maximizing resource utilization efficiency of best-effort traffic.
DMS also exhibits significant advantages over existing schemes in terms of efficiency, complexity, fairness, and throughput. In addition, the comparison with existing power control schemes reveals that complex approaches (e.g.,~\cite{shroff_ubpc}) bring little additional gain with respect to DMS and behave less fairly, whereas low-complexity solutions (e.g.,~\cite{refim}) exhibit lower efficiency.

In the following, we formulate the guaranteed traffic and best-effort traffic problems from a centralized viewpoint in Section~\ref{s:system_model}. In Section~\ref{s:distr_inel} and Section~\ref{s:distr_elastic} we define a distributed version of such problems, and use game theory to analyze them. In Section~\ref{s:framework_dms} we propose DMS, a unified framework meant to supervise the distributed solution of the two distributed problems. DMS is validated in Section~\ref{s:results}. Section~\ref{s:related} provides a thorough report on the current literature and Section~\ref{s:conclusions} concludes the work with some final remarks.

\vspace{-2mm}
\section{System model and centralized problems}
\label{s:system_model}
The goal of ICIC algorithms is to improve system spectral efficiency by  controlling base station mutual interference so that transmissions can be performed with high rate modulation and coding schemes.
To this end, 3GPP has defined the ABSF mechanism, where base stations are instructed to remain silent over some periods in order to avoid interfering with each other, and thus, harming overall performance. Specifically, ABSF orchestrates base stations activities by performing scheduling on a time-slot basis, i.e., per Transmission Time Interval (TTI), and forcing base stations to be silent in some selected TTIs. We say that such base stations are blanked, and refer to the overall schedule of base stations as ABSF time-patterns. 

In the following, we formulate the ICIC problem wherein the ABSF standard technique is implemented. A problem solution consists in a set of  {\it ABSF time-patterns}, i.e., bitmaps that specify which TTIs must be blanked, to be assigned to base stations.\footnote{For the sake of simplicity, problem formulations presented in this section consider downlink traffic only. However, an extended model can be derived for uplink transmissions rather easily.} Our network model addresses two distinct traffic classes: ($i$) guaranteed bit-rate (GBR) and ($ii$) best-effort. While the former is subject to a strict rate constraints and it is accommodated with higher priority, the latter can be served with the remaining resources since it has no stringent requirements in terms of latency and bandwidth.

We tackle the above problem by observing that it can be solved in two steps, due to the fact that the GBR traffic is inelastic while the best-effort one is elastic:
($i$) first, one can find a global time-allocation strategy for different base stations able to accommodate GBR traffic demands into a minimum number of TTIs, ($ii$) then, use the TTIs left to serve best-effort traffic while maximizing the network spectral efficiency and guaranteeing a good level of fairness. 

In what follows, we first formulate the ICIC problem from a \emph{centralized} scheduling perspective for both traffic types. 
While such a solution is practically infeasible due to computational complexity and signaling overhead, it offers
a benchmark corresponding to the best possible performance of any algorithm. 

\subsection{Optimizing GBR Traffic Period}
\label{s:inelastic_problem_centr}

We formalize the problem of optimizing the GBR traffic period length as follows. 
Let us assume that the network consists of a set $\mathcal{N}$ of base stations, each of which having a set of users  
$\mathcal{U}_i$, so that $\mathcal{U} = \cup_{i\in \mathcal{N}} \mathcal{U}_i$ is the set of users in the network. 
Let the GBR traffic demand of each user be known at the base station side, expressed in volume of traffic to be periodically served, and denoted as $D_u, u \in \mathcal{U}_i$. Let $W$ denote the available {\it time horizon} (in TTIs), i.e., the length of ABSF patterns, which means that the user demand guaranteed rate is $D_u / (W \cdot T_{slot})$ bps, where $T_{slot}$ is the duration of a single TTI.
Let us further assume that a base station can schedule at most one user in each TTI, and some TTIs can be~\emph{blanked} by means of the ABSF pattern. While this assumption helps to keep tractable our problem formulation, it can be readily relaxed taking into account multiple users scheduled within the same TTI.

The {\it GBR traffic period} devoted to serve the GBR traffic demands will be no longer than a given portion of the $W$ TTIs; without loss of generality, let us assume that this period is a set of consecutive TTIs, $\mathcal{T}= \{1,2, \dots , W\}$. 
The objective of the optimization problem is to allocate user demands in the smallest number of TTIs, $L \leq W$, while satisfying channel quality constraints. 
The remaining $W-L$ TTIs can be used for best-effort traffic allocation. 

Since the system has limited capacity, the above problem may not be feasible as it may not be possible to allocate the entire user demand set within the assigned slots. In order to ensure that the problem is always mathematically feasible, we define a per-user {\it penalty} $p_u$ representing the unserved demand. As long as no penalty is accumulated, the solution minimizes the GBR traffic period $L$ leaving more room (e.g., $W\!-\!L$ TTIs) for best-effort traffic.

We formulate the optimization problem for the GBR traffic period by considering an SINR-aware scheduling of users with interference thresholds and penalties.
We denote by $\alpha>0 $ the relative importance of penalties $p_u$ over utilized TTIs and by $s_t$ the binary variables indicating whether TTI $t$ is used for transmissions by at least one base station. 
Similarly, binary variables $y_{i,t}$ indicate whether BS $i$ uses TTI $t$ and binary variables $x_{u,t}^m$ indicate whether user $u$ is scheduled in TTI $t$ with modulation and coding scheme (MCS) $m \in \mathcal{M}$, in which case it receives a rate $R^m \in \mathcal{R}$, in bits per-TTI.
We use $L$ to indicate the highest index of used TTIs within $\mathcal{T}$. 
As concerns the rate assigned to users, we consider the signal-to-interference-plus-noise ratio (SINR) computed with 
base station transmission power $P$\,\footnote Following current cellular deployments, we assume that base stations transmit at a constant power. However, our solution can be easily extended to heterogeneous deployments wherein different power levels are set without affecting the overall system performance., 
channel gain $G_{u,k}$ between user $u$ and the base station $k$, 
and background noise $N_0$.
The use of transmission rate $R^m$ is subject to the availability of a SINR value greater than a corresponding threshold $\gamma^{m}$ (see, e.g., \cite{LTEadv36213}).
With the above definitions, the problem can be formally defined as follows:

\vspace{1mm}
\noindent \textbf{Problem} \texttt{GBR}:
\begin{equation*}
\label{pr:gbr}
\begin{array}{ll}
\text{minimize }  & L + \alpha \sum\limits_{u \in \mathcal{U}} p_u, \\
\text{subject to } & t s_t \leq L, \forall t \in \mathcal{T};\\
       & \sum\limits_{i \in \mathcal{N}} y_{i,t} \leq N s_{t}, \forall t \in \mathcal{T}; \\
			 & \sum\limits_{u \in \mathcal{U}_i, m \in \mathcal{M}} x_{u,t}^m \leq y_{i,t}, \forall i \in \mathcal{N}, t \in \mathcal{T};\\
			 & \sum\limits_{u \in \mathcal{U}, m \in \mathcal{M}} x_{u,t}^m \leq 1,  \forall t \in \mathcal{T};\\
			 & \frac{P\,G_{u,i}}{N_0 + \sum\limits_{k \in \mathcal{N}: k \neq i} P\,G_{u,k} \cdot y_{k,t}} \geq \gamma^m \cdot x_{u,t}^m,\\
			 & \quad \qquad \forall i \in \mathcal{N}, u \in \mathcal{U},t \in \mathcal{T},m \in \mathcal{M}; \\
			 & \sum\limits_{m \in \mathcal{M}, t \in \mathcal{T}} R^m x_{u,t}^m + p_u \geq D_u, \forall u \in \mathcal{U};\\
			 & s_t, x_{u,t}^m \in \{0;1\}, \forall u \in \mathcal{U}, m \in \mathcal{M}, t \in \mathcal{T};\\
			 & y_{i,t} \in \{0;1\}, \forall i \in \mathcal{N}, t \in \mathcal{T};\\
			 & p_u \geq 0.
\end{array}
\end{equation*}

The first set of constraints forces the correct value to be  assigned to $L$. The second and third sets of constraints impose the coherence between, respectively, ($i$) active BSs and used TTIs, ($ii$) scheduled users and active BSs. The fourth set of constraints imposes that at most one user may be scheduled in each TTI. 
The fifth set of constraints is used to match SINR and used rates. Although SINR constraints are not linear, they can be easily linearized. 
The sixth set of constraints expresses the target per-user volume of data, while the remaining constraints are used to define the ranges of the decision variables. In particular, the last constraint is used to set penalty values in order to compensate eventually unserved traffic demands. Problem \texttt{GBR} can be solved with state-of-the-art Mixed-Integer Linear Programming (MILP) solvers.

The main assumption behind the centralized model is that users' CSI is perfectly known. Such information is gathered and updated by a centralized optimizer, which uses it to compute SINR constraints.

Note that Problem~\texttt{GBR} can be abstracted by considering $t$ as a transmission opportunity rather than  a TTI. For instance, $t$ can be considered as an available physical resource block (PRB) serving multiple users.

\vspace{-4mm}
\subsection{Optimizing Best-effort Traffic Period}
\vspace{-0.5mm}
\label{s:problem}
Once a feasible GBR traffic period $L$ is found, the remaining $Z\!=\!W\!-\!L$ TTIs in the ABSF pattern will be used for accommodating best-effort traffic demands. Differently from the GBR case, here the goal is to obtain a user scheduling and BS activation that can efficiently exploit the remaining network resources by aiming at both spectral efficiency and  user fairness. We can formulate the optimization problem with an  Integer Linear Programming (ILP) model. The objective function to be maximized, $\widehat{\eta}$, is the sum of the utilities of the individual base stations. Following the 
\emph{max-min} fairness criterion, we define the utility of base station $i$ as the minimum rate of all its users and formalize the problem as follows:%
\footnote{The selected objective function provides a trade-off between maximizing the spectral efficiency and guaranteeing a minimum level of service quality.
Nevertheless, using a different objective function would not substantially change the proposed approach and the following analysis.}

\vspace{1mm}
\noindent \textbf{Problem} \texttt{BE}:%
\vspace{-2mm}
\begin{equation*}
\label{pr:be}
\begin{array}{ll}
\!\! \!\!\text{maximize} & \!\!\widehat{\eta} =\sum\limits_{i \in \mathcal{N}} \left( \min\limits_{(u,t) \in \mathcal{U}_i \times \mathcal{Z}} \sum\limits_{m \in \mathcal{M}}R^m \cdot x_u^{m,t} \right), \\
\!\! \text{subject to } & \!\! \sum\limits_{u \in \mathcal{U}_i,m \in \mathcal{M}} x_u^{m,t} \leq y_{i,t}, \quad \forall i \in \mathcal{N},t \in \mathcal{Z},\\
\!\! 	 & \! \! \frac{P\,G_{u,i}}{N_0 + \sum\limits_{k \in \mathcal{N}: k \neq i} P\,G_{u,k} \cdot y_{k,t}} \geq \gamma^m \cdot x_u^{m,t}, \\
\!\!      &\quad \qquad \forall i \in \mathcal{N}, u \in \mathcal{U}_i,m\!\! \in \!\! \mathcal{M},t\!\! \in \mathcal{Z},\\
\!\!      &  \!\! \!\!\!\! y_{i,t},x_u^{m,t}\!\! \in\! \{0;1\}, \forall i \!\in \!\mathcal{N}\!, u \! \in \!\mathcal{U}_i, m \!\!\in \!\!\mathcal{M}, t\!\! \in \!\mathcal{Z};
\end{array}
\end{equation*}
where variables and parameters are defined exactly as in Problem~\texttt{GBR}. The set $\mathcal{Z}$ is the set of available TTIs for best-effort traffic and it is defined as $\mathcal{Z}=\left\{1,2...,Z\right\}$. The two sets of constraints correspond to the third and fourth ones in Problem~\texttt{GBR}.
Problem~\texttt{BE} can be reduced to a bin-packing problem in which the sum of interferences cannot exceed a threshold. Therefore, this problem is NP-hard~\cite{goussevskaia:capacity}. 

As stated before, the solution of Problem~\texttt{GBR} and Problem~\texttt{BE}  involves a very high overhead to deliver CSI information to a centralized optimizer, which needs this information to select the ABSF patterns and compute the user scheduling. In addition, due to problem complexity, while the centralized approach can be an attractive option for small networks, a computationally less complex approach is required to deal with the case of very dense wireless networks consisting of hundreds of base stations and thousands of wireless nodes. 

Accordingly, in the following section, we present a distributed and less complex approach to this joint problem in order to abate and distribute the computational load over the base stations. After analyzing the two problems individually, we then propose a joint framework for both problems.

\section{Guaranteed traffic}
\vspace{-1mm}
\label{s:distr_inel}
We formulate the GBR problem in a distributed way by splitting it into local problems that are solved by each base station. 
To reduce complexity, each base station only optimizes the scheduling of its own users and considers that other base stations use fixed ABSF patterns. However, this approach needs an iterative mechanism.

Each local optimization problem consists in minimizing a cost function $f_i$ that accounts for both the number of locally used TTIs and the total penalty related to unsatisfied local demands.
Variables and constraints are the same as in the centralized formulation, except for two aspects: ($i$) the local formulation considers only the users of the local base station, ($ii$) the activity of interfering stations in the SINR constraint is no longer optimized, but given as input. 
The key idea behind the distributed approach indeed affects only this constraint. In fact, the exact knowledge of which users are scheduled by other base stations is no  needed, as the base station activity is enough to compute SINR values. Therefore, it is sufficient to know the activity patterns of neighboring base stations, namely ABSF patterns, provided by binary vectors $\{A_t^k\}$. The formal description of the distributed problem is as follows: 

\vspace{1mm}
\noindent \textbf{Problem} \texttt{GBR-DISTR} (local, at BS $i$):
\begin{equation*}
\label{pr:gbr-distr}
\begin{array}{ll}
\text{minimize }  & f_i=\sum\limits_{u\in \mathcal{U}_i,t \in \mathcal{T}} x_{u,t}^m + \alpha\,\sum\limits_{u \in \mathcal{U}_i}p_u, \\
\text{subject to } & \sum\limits_{t \in \mathcal{T},m \in \mathcal{M}} R^m \cdot x_{u,t}^m + p_u \geq D_{u}, \quad \forall u \in \mathcal{U}_i; \\
				   & \sum\limits_{u \in \mathcal{U}_i, m \in \mathcal{M}} x_{u,t}^m \leq 1, \quad \forall t \in \mathcal{T}; \\
				   & \frac{P\,G_{u,b}}{N_0 + \sum\limits_{k \in \mathcal{N} \setminus {i}} P\,G_{u,k} \cdot A_t^k} \geq \gamma^m \cdot x_{u,t}^m, \\
				   & \quad \qquad \qquad \qquad \forall u \in \mathcal{U}_i,t \in \mathcal{T},m \in \mathcal{M}; \\
				   & x_{u,t}^m \in \{0;1\}, \quad \forall u \in \mathcal{U}_i, t \in \mathcal{T}, m \in \mathcal{M}; \\ 
				   & p_u \geq 0, \qquad \forall u \in \mathcal{U}_i.
\end{array}
\end{equation*}

Each base station $i$ is in charge of solving Problem~\texttt{GBR-DISTR} by computing the optimal user scheduling into available TTIs. Note that the solution of this problem depends on the solutions computed by the other base stations, since the SINR of each user is affected by the interference generated by the other base stations when they are active. Therefore, an iterative process is needed, which would ideally converge to a quasi-optimal solution where base stations agree on their respective ABSF patterns. However, the process could not converge at all. 

We now derive convergence properties and provide conditions on the guaranteed convergence by casting the distributed approach into a game.
Once the convergence is guaranteed, we finally present a practical distributed scheme that implements the distributed approach.
\label{s:game_inel}

{\bf Game Theoretical Analysis.} We introduce a new class of games, called \emph{Distributed Inelastic Games}, to model the interference coordination problem. 

The game that describes our problem is represented by means of a tuple $\Gamma \!=\! (\mathcal{N}, (\mathbb{S}_i)_{i\in \mathcal{N}}, (f_i)_{i\in \mathcal{N}})$.  The set of players is $\mathcal{N}$, i.e., the  base stations. For each player $i \in \mathcal{N}$, $\mathbb{S}_i$ is a family of user actions, namely a {\it strategy}, and $f_i$ is a cost function that expresses the cost associated to the implementation of each action, i.e., $f_i$ is the utility function in Problem \texttt{GBR-DISTR}. 
Each player $i$ decides her action in order to minimize the game cost function $f_i$.
Player $i$'s action is a set of user-TTI pairs $(u,t)$, which represents the base station's user scheduling in the ABSF time horizon. 
In particular, a valid action is a scheduling that satisfies the constraints of Problem~\texttt{GBR-DISTR}. This defines the action space $\mathbb{S}_i$ from which user $i$ selects her action $S_i$.
Note that the cost of each action depends on the other players' actions, because of interference, so we use the notation $f_i(S_i, S_{-i})$ to indicate the dependency on the action of $i$ as well as on the actions of any other user, $S_{-i}$.

With the above, the {\it best response} (BR) for game $\Gamma$ is defined as the action that produces the smallest cost function value for player $i$, taking the other players' actions as given. Analytically, $S^*_i \! \in \! \mathbb{S}_i $ is defined as BR if and only if
\begin{equation}
f(S^*_i, S_{-i}) \leq f(S_i, S_{-i}), ~ \forall S_i \in \mathbb{S}_i.
\end{equation}

In the following, we present a convergence analysis of game $\Gamma$, which is essential to ensure the feasibility and implementability of the distributed version. Indeed, due to the nature of the game, the arbitrary best responses taken by each player may not necessarily lead to an equilibrium (i.e., a Nash equilibrium); this is the case of game $\Gamma$ and it is shown in the proof of Theorem~\ref{theo:oscillations}.
\begin{theorem}
\label{theo:oscillations}
Distributed Inelastic Game $\Gamma$ does not possess a finite improvement property in best-response improvement dynamics.
\end{theorem}
\begin{sketch}
Consider a scenario with $2$ TTIs and $3$ base stations, each of them associated with exactly $ 1$ user. 
For each player $i$, the only user can be scheduled in one or both TTIs, so the valid action space $\mathbb{S}_i$ is defined as follows:
\vspace{-2mm}
\begin{align*}
&\mathbb{S}_1 = \{ \{u_1,t_1\};\{u_1,t_2\};\{(u_1,t_1),(u_1,t_2)\} \}, \\
&\mathbb{S}_2 = \{ \{u_2,t_1\};\{u_2,t_2\};\{(u_2,t_1),(u_2,t_2)\} \}, \\
&\mathbb{S}_3 = \{ \{u_3,t_1\};\{u_3,t_2\};\{(u_3,t_1),(u_3,t_2)\} \}.
\end{align*}
\vspace{-2mm}

Let assume a traffic demand $D_u = 5$ units per ABSF cycle  (any unit can be used, e.g., Kbytes)  and in a TTI a user obtains the number of traffic units described in the following table:

\begin{center}
  \begin{tabular}[c]{| l || c | c | c | c |}        
    \hline
    		& alone	& with $\{u_{i+1}\}$	& with $\{u_{i+2}\}$	& all		\\ \hline
    $u_i$	& $5.55$	& $5.11$			& $2.73$ 			& $2.51$	\\ \hline
  \end{tabular}
\end{center}
In the table above, $i \in \{1,2,3\}$, and user indexes $x=i+1$ or $x=i+2$ have to be computed as shifts on a  cyclically extended list of indexes. In this example, we also assume a penalty $p_u=0.1$ with $\alpha=1000$.

Now we consider the sequence of actions taken by each player, described in Table~\ref{tab:oscillations}.
Whenever a player $i$ chooses a new action at step $k$, she uses the best response at that step, which affects the traffic received by other users. In particular, player $i$ causes player $i+1$ to incur a penalty because of low throughput ($2.73$ units in a TTI and $0$ in the other). So, the affected player will switch scheduling TTI and get $5.11$ units of traffic, which is enough to satisfy $D_u$ with a single TTI, although this action will trigger the next player in the list to do the same.  
This leads to a cycle of equal actions, such as actions at step $k$ and actions at step $k+6$. Hence, playing arbitrary best responses do not necessarily converges to a Nash equilibrium
in distributed inelastic games.
\qed
\end{sketch}

\begin{table*}[!ht]
	\footnotesize
  	\centering    
    \caption{Example of dynamics of game states for a Distributed Inelastic Game $\Gamma$ by adopting Best Response. Base stations play sequentially: moves are indicated by using boldface fonts for the corresponding scheduling strategy and cost function.}
    \vspace{-2mm}
  \begin{tabular}{|c|c|c|c|c|c|c|}
	\hline   
     & step $k-1$ & step $k$ & step $k+1$ &step $k+2$ &  \dots & step $k+6$\\
    \hline   
    BS 1 & $\{u_1,t_1\}, f_1=1$ & $\{u_1,t_1\}, f_1=\textit{101}$ & $\boldsymbol{\{u_1,t_2\}}, \boldsymbol{f_1=1}$ & $\{u_1,t_2\}, f_1=1$ & \dots & $\{u_1,t_1\}, f_1=\textit{101}$ \\      
    BS 2 & $\{u_2,t_2\}, f_2=1$ & $\{u_2,t_2\}, f_2=1$ & $\{u_2,t_2\}, f_2=\textit{101}$ &  $\boldsymbol{\{u_2,t_1\}}, \boldsymbol{f_2=1}$ & \dots & $\{u_2,t_2\}, f_2=1$ \\ 
    BS 3 & -  & $\boldsymbol{\{u_3,t_1\}}, \boldsymbol{f_3=1}$ & $\{u_3,t_1\}, f_3=1$ & $\{u_3,t_1\}, f_3=\textit{101}$ &  \dots & $\boldsymbol{\{u_3,t_1\}}, \boldsymbol{f_3=1}$ \\
    \hline
  \end{tabular}
  \label{tab:oscillations}
  \vspace{-5mm}
\end{table*}

While the asymmetric scenario considered in the above proof is quite unlikely in realistic LTE-Advanced environments, the theorem does nonetheless point out that the game $\Gamma$ may not converge in some critical scenarios, which may create applicability problem if not properly fixed.

However, using a particular class of best responses leads to an equilibrium. This particular best response set consists in selecting among all possible best responses only those which just add or remove at most one $(u,t)$ pair, to the action of the previous step. We call such set as \emph{Single-step} set, $\mathbb{S}_i^{SS}$, and define it formally as follows. Starting from any action $S_i^{(p)}$ taken at the previous step $p$, $\mathbb{S}_i^{SS}(S_i^{(p)}) = \{ S \in \mathbb{S}_i : (|S \setminus S_i^{(p)}| \leq 1) ~ \vee ~ (|S_i^{(p)} \setminus S| \leq 1) \}$, where the $\vee$ symbol is the OR operator. Now we can define a Single-step Best Response (SSBR) move:
\begin{definition}[SSBR]
At step $k$, the Single-step best response (SSBR) $\hat{S_i}^{k}$ of player $i$ is defined as a best response action $S_i^{*(k)}$ such that $S_i^{*(k)} \in \mathbb{S}_i^{SS}(S_i^{(k-1)})$.
\end{definition}

The above definition states that player $i$ will play her single-step best response by taking into consideration her action played at the previous step and ($i$) removing one of the (user,TTI) pair, ($ii$) adding just one additional (user,TTI) pair, or ($iii$) following the previous action (if the cost function is minimized for that particular action). 

In order to prove that the convergence is guaranteed by following the SSBR approach, we next introduce the concept of action profile. Given a state of the game $\Gamma$ at a particular round, the action profile $\sigma$ is the set of actions played by each player in that round. When a player changes her action, the action profile is updated. In the following, we define a particular action profile, namely \emph{saturation action profile}.
\begin{definition}
The saturation action profile is defined as an action profile $\overline{\sigma} = [S_1,...,S_N]$ belonging to a set of saturation action profiles $\Sigma^{SAT}$, $\overline{\sigma} \in \Sigma^{SAT}$, where each player's action $S_i$ either ($c.i$) returns a cost function with a zero penalty, or ($c.ii$) occupies all available $W$ TTIs with a non-zero penalty, or ($c.iii$) has a non-zero penalty and does not occupy every slot, but free slots provide a null contribution to the cost function.
\end{definition}

Assuming that the players play their SSBR for a game $\Gamma$, we can now formulate the following lemmas, whose formal proofs are reported in the Appendix.

\begin{lemma}\label{Theo:convergence}
Given that the players' actions belong to whatever action profile $\sigma$, after a finite number of single-step best responses, all players' actions will belong to a saturation action profile $\overline{\sigma}$.
\end{lemma}

\begin{lemma}\label{Theo:sat_point}
At a certain point in time, given that the actions selected by any player in the system belong to a saturation action profile $\overline{\sigma}$, 
if each player chooses a single-step best response, the game will converge to a Nash equilibrium.
\end{lemma}

Relying on such Lemmas, we can prove the following result. 

\begin{theorem}\label{Theo:nash}
Game $\Gamma$ possesses at least one Nash equilibrium and players reach an equilibrium after a finite number of single-step best responses.
\end{theorem}
\begin{proof}
We prove it by a constructive proof. Players start playing a game $\Gamma$. Regardless of the starting action profile $\sigma$, after playing a finite number of SSBR, the players' actions belong to a saturation action profile $\overline{\sigma}$, as stated by Lemma~\ref{Theo:convergence}. Upon all players select an action belonging to a saturation action profile, keeping choosing a SSBR, they will converge in a finite number of steps to a Nash equilibrium according to Lemma~\ref{Theo:sat_point}. Therefore, we can state that each game $\Gamma$ admits a Nash equilibrium, and players can reach such equilibrium. 
\end{proof}

The proof of the theorem is also confirmed by readily applying the SSBR to the scenario presented in the proof of Theorem~\ref{theo:oscillations}.
In that example, choosing the SSBR for all players leads to fully schedule all available TTIs for every base station.

Interestingly enough, Theorem~\ref{Theo:nash} also proves that players can easily adopt a general best response strategy during the game, with no convergence guarantees. However, if at a certain point in time, they switch to SSBR strategy, they converge to a Nash equilibrium with probability equal to $1$. Clearly, if at least one player is not playing SSBR, the game convergence is no longer guaranteed.

\section{Best-Effort traffic}
\label{s:distr_elastic}
We next present a distributed formulation of Problem~\texttt{BE}.
Similarly to what presented in the previous section for the inelastic traffic, here we formulate a distributed approach for the scheduling of elastic traffic. Also in this case, the original problem is split into several smaller instances, which are solved locally by each base station. To solve a local problem instance, a base station is provided with the activity pattern 
declared by other base stations. 

With the above information, and without explicitly forcing any additional constraint, each base station $i$ would schedule users selfishly in the entire set of $Z$ TTIs, in order to optimize the local utility. Therefore, to avoid that base stations use all available TTIs, in the distributed problem formulation we grant a single base station $i$ access to up to $M_i$ TTIs over $Z$ available TTIs; such $M_i$ value plays a key role in the distributed mechanism, as it will be clarified in Section~\ref{s:design}. 

The above description corresponds to the following instance of the local problem for base station $i$, which can be formulated as an ILP model as follows:

\vspace{1mm}
\noindent \textbf{Problem} \texttt{BE-DISTR}:
\begin{equation*}
\label{pr:be-distr}
\begin{array}{ll}
\text{maximize}  & \widehat{\eta_i} =\min\limits_{(u,t) \in \mathcal{U}_i \times \mathcal{Z}} \sum\limits_{m \in \mathcal{M}}R^m \cdot x_u^{m,t}, \\
\text{subject to} & \sum\limits_{u \in \mathcal{U}_i,m \in \mathcal{M}} x_u^{m,t} \leq 1, \qquad \forall t \in \mathcal{Z},\\
				   & \frac{P\,G_{u,i}}{N_0 + \sum\limits_{k \in \mathcal{N}: k \neq i} P\,G_{u,k} \cdot A^k_t} \geq \gamma^m \cdot x_u^{m,t},\\
				   & \hfill \forall u \in \mathcal{U}_i, m \in \mathcal{M}, t \in \mathcal{Z},\\
				   & \sum\limits_{u \in \mathcal{U}_i, m \in \mathcal{M}, t \in \mathcal{Z}} x_u^{m,t} \leq M_i,\\
				   & x_u^{m,t} \in \{0;1\}, \quad \forall u \in \mathcal{U}_i, m \in \mathcal{M}, t \in \mathcal{Z};
\end{array}
\end{equation*}
where all parameters and constraints have the same meaning as in Problem~\texttt{BE}, except for the third constraint, which limits the number of usable TTIs to $M_i$. Note that a feasible solution of Problem~\texttt{BE-DISTR} can be computed by using any available max-min scheduling heuristic (see, e.g.,~\cite{Bertsekas}).

Following the local optimization problem presented above, the interference coordination problem is solved in a distributed fashion: each base station receives as input the ABSF patterns,
solves Problem~\texttt{BE-DISTR} and provides in turn to other base stations its ABSF pattern. Other base stations update their choices depending on the new ABSF pattern, communicate back their new ABSF decisions and the process repeats. However, this process may not converge. 

{\bf Game Theoretical Analysis.} 
The fully distributed approach to best effort traffic scheduling can be analyzed by means of well-known results offered by game theory. Specifically, the distributed approach described above can be modeled as a game where base stations iteratively play in order to maximize their utility. Let us define this game as an \emph{Interference Coordination Game} $\Omega$, where each base station $i$ acts as a player. Similar to game $\Gamma$, the set of actions of each player $\mathbb{S}_i$ consists in the set of pairs (user, TTI), $(u,t) \in \mathcal{U}_i \times \mathcal{T}$, available for each base station according to constraints in Problem~\texttt{BE-DISTR}. However, the cost functions in the two problems are different.

In order to analyze the convergence of the game, we rely on the concept of \emph{Bottleneck Matroid Congestion Game} (for further details we refer the reader to~\cite{harks2013bottleneck}). 
The latter is a class of games in which resources are shared among players. The utility of each player depends on the utility of the resources she chooses and the number of players choosing the same resources: the higher the congestion, the lower the utility. In particular, the individual player utility is the minimum of the utilities of the resources chosen in her action.
Note that the bottleneck behavior comes from the max-min objective function of Problem~\texttt{BR-DISTR}.

Congestion games can be generalized in \emph{player-specific congestion games} and \emph{weighted congestion games}. In the former, every player has her own utility function for every resource. In a weighted congestion game, every player affects the other players strategies with a different weight, namely, she causes a different level of congestion. 
A thorough game-theoretical analysis of Game $\Omega$ is not a specific contribution of this work, since we have already provided it
in~\cite{secon15_sfmcb}.
Here we simply recall the relevant results of~\cite{secon15_sfmcb}, i.e.:
$(i)$
the Interference Coordination Game $\Omega$ is a Weighted Player-specific Bottleneck Matroid Congestion Game,
and $(ii)$ weighted player-specific matroid bottleneck congestion games do not exhibit the finite improvement property  when  using the best response. 
Therefore the distributed approach may not converge. Moreover, the solution of Problem~\texttt{BE-DISTR} is biased by the $M_i$ values.

In order to address these shortcomings, in the next section we propose a {\it semi}-distributed two-level mechanism where an overall supervisor guides the behavior of the distributed game,  using only resources not otherwise allocated to GBR traffic.

\section{Distributed Multi-traffic Scheduling Framework}
\label{s:framework_dms}
In this section we provide complete details on the proposed framework for adaptive interference-aware scheduling.
We name our scheme \emph{Distributed Multi-traffic Scheduling (DMS)}. DMS is based on the game theoretical framework introduced in Section~\ref{s:game_inel} and Section~\ref{s:distr_elastic}, and incorporates heuristic approaches to jointly adapt the solutions of Problem~\texttt{GBR} and Problem~\texttt{BE} to traffic changes. 
To cope with traffic and network dynamics, a practical strategy consists in periodic scheduling decisions, taken once per time horizon $W$ (e.g., every ABSF pattern), to serve inelastic traffic demands and best-effort data traffic within the ABSF pattern $W$. DMS first accommodates inelastic traffic that exhibits very stringent requirements. The remaining time portion is left for best-effort traffic requests. 

The first objective of DMS is to smartly optimize the inelastic traffic period scheduling in order to maximize the resource efficiency, while leaving more space for best-effort traffic. To this aim, DMS includes a mechanism that adaptively seeks the minimal number of TTIs to include in the inelastic traffic period $T \le W$, so as inelastic traffic is served with no penalties in the shortest possible time window. Without loss of generality, we only consider scenarios where the inelastic traffic can always be served within the W horizon, otherwise no solution could be found. Indeed, it would mean that the set of inelastic traffic demands cannot be entirely guaranteed and some of them must be reduced or rejected.

In the following we show that, with the game theoretical approach proposed, a lightweight supervisor suffices to solve the inelastic period minimization problem by leveraging a simple dichotomic search algorithm over several time horizons. Then, DMS fully exploits the remaining time portion (e.g., best-effort traffic period $W\!-\!T\!=\!Z$) to maximize the aggregate system throughput for serving best-effort traffic. This is automatically performed using a distributed mechanism without incurring in a heavy centralized channel statistics collection.
Please note that our solution is practical and implementable: the local controller is developed on the base station, e.g. the Radio Network Controller (RNC) for UMTS architecture or E-UTRAN in the LTE architecture, whereas the supervisor can be envisioned as an SDN-based controller in charge of collecting channel information and commuting simple decisions (\cite{LMR_EWSDN12}) or in Cloud-RAN networks it can be integrated in the BaseBand Unit (BBU). This is pretty in line with future networks design, wherein Software Defined Network (SDN) paradigm is fully applied to mobile networks~\cite{5gnorma}.

\subsection{Guaranteed traffic scheduling}
\label{s:framework_inel}

For inelastic traffic demands, DMS focuses on two different objectives. While the first objective is to fully accommodate the guaranteed traffic demands into the available time horizon $T$ ({\it Resource Allocation}), the second objective is to iteratively reduce the number of used TTIs in order to make an efficient use of the time resources ({\it Time Squeezing}). 

\begin{itemize}
\item{\it Resources Allocation} is completely executed into base stations, each of which is in charge of jointly scheduling local users and making ABSF pattern decisions, which are exchanged with the other base stations through the high-level supervisor (HS). Game $\Gamma$ is used for the base stations to accomplish this task in a coordinated way.
\item{\it Time Squeezing} is executed at the HS. The HS collects the traffic demand offered to the base stations and iteratively adjusts the length of the time period $T$ based on the ABSF patterns announced by the base stations at the end of game $\Gamma$, and on penalties they could have incurred.   
\end{itemize}
Practically speaking, DMS operation starts when user traffic demand changes in the cellular network, as illustrated in Fig.~\ref{fig:flowchart}. 
Initially, each base station provides the HS with its cumulative traffic demand.
The HS selects the initial time period as the one that would guarantee traffic constraints without adopting any interference coordination mechanism ($T = W$). As a consequence, the computation of assigned time resources is initially overestimated.
After that, DMS operation consists in the interaction between the two aforementioned processes: the Resources Allocation process and the Time Squeezing process.

{\bf Resources Allocation process: guaranteeing user demands.}
During Resource Allocation, the number of available TTIs $T$ is fixed.
Base stations cooperatively schedule their own users into available TTIs in order to satisfy their traffic demands. It is very important to note that the mechanism perfectly complies with the requirements of the Distributed Inelastic Game 
$\Gamma$ presented in Section~\ref{s:distr_inel}. In particular, each base station schedules its own users and communicates corresponding ABSF pattern to the other base stations. 
Each base station limits its activity and reduces the interference caused to the other base stations by reducing the number of occupied TTIs, as stated in Problem~\texttt{GBR-DISTR}. 
The process ends when a steady-state is reached, which always occurs if SSBR is enforced, as proved in Section~\ref{s:distr_inel}. 
Eventually, a steady state ABSF pattern for each involved base station is notified to the HS. In addition, the Resource Allocation process may output a set of non-zero user penalties, e.g., due to a too large traffic to be accommodated in the time horizon or to critical interference conditions. This event is promptly handled by the Time Squeezing process described next.

\begin{figure}[!t]
\centering
	\includegraphics[trim= 7mm 139mm 5mm 20mm, clip=true, width=0.55\textwidth, angle=90]{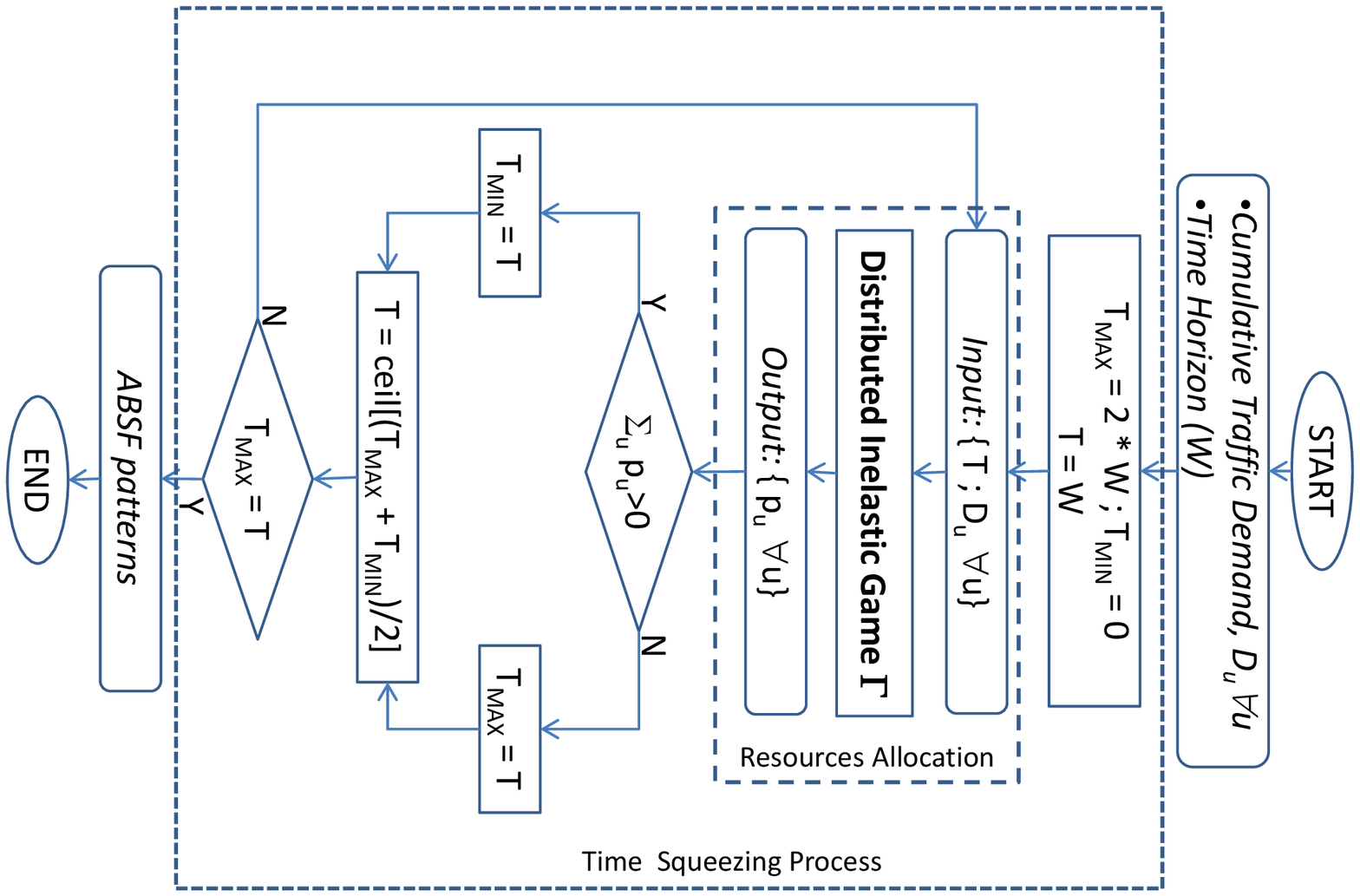}
	\caption{\footnotesize Two-level mechanism for guaranteed traffic. In the flow chart, the entire flow of the program is represented as an iteration between the execution of the Resources Allocation and of the Time Squeezing processes.}	
	\label{fig:flowchart}	
\end{figure}

{\bf Time Squeezing: adapt time period to demand.}
Time Squeezing is based on a binary search scheme, as illustrated in Fig.~\ref{fig:flowchart}. 
The initial time period $T$ is set equal to the ABSF pattern length $W$, chosen as the number of TTIs needed to guarantee the entire demand. 
Then, a binary search is used to adapt the inelastic time period. At each step of the search (e.g., at the beginning of each ABSF pattern $W$), a new value is chosen for the time period 
and it is applied by invoking the Resource Allocation process. The Resource Allocation process runs and returns ABSF patterns and penalties. If the sum of the obtained penalties is equal to zero, meaning that user traffic demands are completely satisfied and the time period may not be fully utilized, Time Squeezing reduces the time period $T$ for the next ABSF pattern. If penalties occur, the process increases the time period. The search ends if there are no penalties and no unused TTIs, or after $\log_2 W$ steps:
base stations keep using the same inelastic traffic period for next ABSF patterns, unless traffic demand changes. 

Since during the Time Squeezing process a feasible value (no penalties) of the time period $T$ is always available, the supervisor can command the base stations to apply the corresponding ABSF patterns and transmission scheduling without waiting for the convergence of the process.
Hence, although with the first applied ABSF patterns resources are not used in  a perfectly efficient way since $T$ is larger than necessary, the system is always able to guarantee the rates of inelastic demand.  However, as we will show in Section~\ref{s:perf_inelastic}, in practical scenarios it takes only a few iteration for DMS to find efficient time periods. 

\subsection{Best-effort traffic scheduling }
\label{s:design}

Building on the results of the previous section, DMS exploits the best-effort traffic period $Z$ to maximize the overall system throughput for serving best-effort traffic requests.
%
\begin{figure}[!t]
\centering
	\includegraphics[scale=0.35]{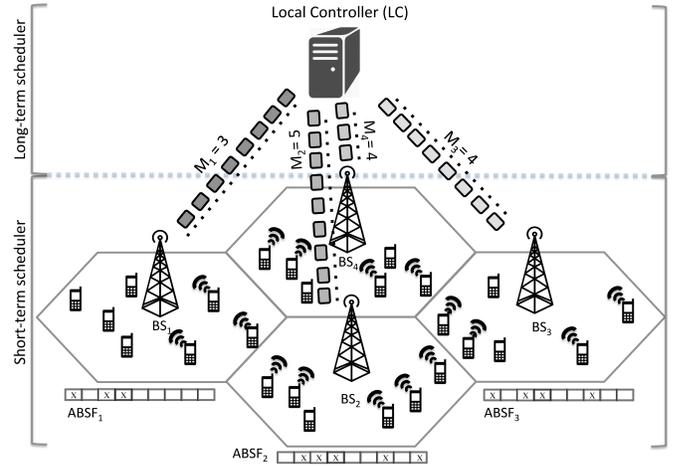}
	\caption{\footnotesize Hybrid two-level mechanism for best-effort traffic demands. In the short-term level (bottom side of the figure), game $\Omega$ is played amongst the base stations, while in the long-term level (top side) the high-level supervisor decides the number of available TTIs per base station.}	
	\label{fig:hyicic}	
\end{figure}
The DMS mechanism for best-effort traffic is depicted in Fig.~\ref{fig:hyicic}. As shown, the scheme operates at two different timescales:
\begin{itemize}
\item On a {\it long-term timescale} (in the order of seconds), HS is in charge of adjusting the $M_i$ value of each base station, where $M_i$ gives the maximum number of TTIs that base station $i$ can use to schedule its users within the time horizon $Z$ by solving Problem~\texttt{BE-DISTR}.
In addition, adapting $M_i$ allows the system to react to traffic changes.
\item At a {\it shorter timescale}, base stations play the Interference Coordination Game $\Omega$ by sequentially exchanging their scheduling decisions in terms of ABSF patterns.
\end{itemize}
Note that Game $\Omega$ is played in a sequential manner to provide converge guarantees.
The supervisor does not directly participate in the game, but it controls its convergence by limiting the number of iterations.
The challenge for the coordinator-aided approach is the design of the algorithms executed by the supervisor to ($i$) ensure convergence, and ($ii$) adjust the values $M_i$. In the following we address the design of those algorithms, which aim at driving the system behavior to an optimal state in the long run. 

{\bf Convergence control of game $\Omega$.} In order to guarantee the convergence of the game, the supervisor imposes a deadline of $\hat{Z}$ TTIs, with $\hat{Z} < Z < W$: if the game has not finished by this deadline, it is terminated by the HS.

When the game finishes before the deadline, the resulting scheduling corresponds to an equilibrium of the game, which ensures that resources are fairly shared among base stations. In contrast, when the game is terminated by the supervisor, base stations use the scheduling that they computed in the latest iteration of the game, which does not correspond to an equilibrium. Thus, in the latter case some base stations could potentially have a better scheduling (i.e., more resources) than the others. However, as shown by our results of Section~\ref{s:perf_elastic}, we have observed that in practice the game can be interrupted after only a very few iterations without negatively impacting fairness or performance in a significant manner.

The deadline $\hat{Z}$ is chosen in order to have a valid scheduling before the current period $Z$ finishes: the resulting scheduling (and the corresponding ABSF pattern) will then be used for the next period. During the game, transmissions and users are scheduled according to the result of the previous period. Note that the iterations of game $\Omega$, as it is in game $\Gamma$, do not need to be synchronized with the TTIs; they can be much faster, allowing for more than $Z$ iterations within $Z$ TTIs. Indeed, the execution of one iteration only requires passing the ``current'' ABSF patterns from one base station to another. As shown in Section~\ref{s:perf_elastic}, deadline $\hat{Z}$ can be chosen in the range $[|\mathcal{N}|, |\mathcal{N}|^2]$.


{\bf Dynamic adjustment of TTI bounds $M_i$.}
One critical aspect for the performance of the proposed mechanism is the setting of the $M_i$ parameters, which give the maximum number of non-blank TTIs available to each base station. Indeed, if the $M_i$ values are too small, performance is degraded because, even if base stations can be scheduled one at a time with low interference, the number of TTIs available for transmitting can be too small to accomodate all users. Conversely, if the $M_i$ values are too large, performance is degraded as a result of too many base stations scheduled together and interfering each other. Thus, performance is maximized when the $M_i$ parameters are optimally set to values that are neither too large nor too small. In the rest of this section, we design an adaptive algorithm that follows an \emph{additive-increase multiplicative-decrease} (AIMD) strategy~\cite{aimd_perf} to find the optimal $M_i$ setting.

In addition to optimally setting $M_i$ to improve the performance of the network, the adaptive algorithm also aims at dynamically adjusting the $M_i$ configuration to follow the changes in traffic and interference. From this perspective, the adaptive algorithm is a long-term process.
In contrast, the distributed game is a short-term process played once per each period of $W$ TTIs. This implies that the duration of the period $W$ cannot exceed a few hundreds frames, which corresponds to a few seconds during which traffic and average channel conditions remain practically unchanged.

From a high level perspective, the algorithm works as follows. At the end of each period of $W$ TTIs, i.e., after the BE traffic serving period $Z$, the supervisor gathers from the base stations the performance resulting from the $M_i$ values (and the corresponding ABSF patterns) used during the period. 
The metric used to represent the performance of a base station in terms of elastic traffic is given by the average quantity of traffic served per each of its users in the interval $Z$. If we denote by $c_{u,t}$ the traffic served for user $u$ in TTI $t$, the metric defined above is simply expressed as follows:\footnote{Note that, since user allocation is carried out according to Problem~\texttt{BE-DISTR}, the max-min objective tends to assign rates with limited variance; as a consequence, the average user rate and the rate of the worst-off user are likely to be similar.}
\vspace{-0.5em}
\begin{equation}
\eta_i = \frac{1}{|\mathcal{U}_i|}\sum\limits_{(u,t) \in \mathcal{U}_i \times \mathcal{Z}} c_{u,t} , \quad \forall i \in \mathcal{N}.
\vspace{-2mm}
\end{equation}

The supervisor then uses the sum of the individual performance metrics, $\eta = \sum_{i \in N} \eta_i$, to keep track of the global system performance and drive $M_i$ to the setting that maximizes $\eta$. The algorithm to find such $M_i$ setting follows an AIMD strategy: the $M_i$ values are linearly increased as long as performance is improved, and, when performance stops improving, then the $M_i$ values are decreased multiplicatively. After each update of the $M_i$ values, these are distributed to the base stations and used 
in the following period (i.e., the following iteration of game $\Omega$).

The specific algorithm executed to calculate the new set of TTI bounds $M_i$ is described in Algorithm~\ref{alg:long-term}. Each iteration of the algorithm is identified by an index $k$. At the initial step ($k = 0$), the supervisor initializes the system performance metrics $\eta$ to $0$ and assigns the initial TTI bounds $M^*_i~=~\lceil Z/\vert \mathcal{N} \vert \rceil$ for every base station. This initial $M^*_i$ setting has been chosen to allow base stations to schedule their users in disjoint portions of the period, which helps the convergence of the algorithm in case of very high mutual interference between all base stations. The $M^*_i$ also provides a lower bound for $M_i$.

\begin{algorithm}[!h]
\caption{\small{\textbf{Resource Sharing Algorithm:} Adaptive algorithm to dynamically design $M_i$. Called at the end of $(k-1)^{th}$ ABSF pattern}} 
\label{alg:long-term}
\algsetup{indent=1em}
\begin{footnotesize}
\begin{algorithmic}
\STATE {\bf Input: }$\mathcal{N}, Z, M^*_i, \eta^{(k-1)}$
\STATE {\bf Initialization: } $\eta^{(k)} \leftarrow 0; M_i \leftarrow M^*_i, \forall i \in \mathcal{N} $
\STATE {\bf Procedure}
\end{algorithmic}
\begin{algorithmic}[1]
\STATE $\mathcal{V} \leftarrow \{ \eta_i, \forall i \in \mathcal{N} \}$
\STATE Order $\mathcal{V}$ non-increasing
\STATE $\eta^{(k)} = \sum\limits_{i \in N} \eta_i$

\IF{$\eta^{(k)} > \eta^{(k-1)}$}
    \WHILE{$\mathcal{V} \neq \emptyset$}
    	\STATE e = pop($\mathcal{V}$)
    	\STATE \textbf{Consider} index $i$ of element $e$
    	\IF{$M_i^{(k-1)} < Z$}
    		\STATE $M_i^{(k)} = M_i^{(k-1)} + 1$
    		\STATE \textbf{break}
    	\ENDIF
    \ENDWHILE
\ELSE
    \WHILE{$\mathcal{V} \neq \emptyset$}
    	\STATE e = pop($\mathcal{V}$)
    	\STATE \textbf{Consider} index $i$ of element $e$
    	\IF{$M_i^{(k-1)} > M^*_i$}
    		\STATE $M_i^{(k)} = \max \left\{M^*_i;\left \lceil M_i^{(k-1)}/2 \right \rceil \right\}$
    		\STATE $\eta^{(k)} = 0$
    		\STATE \textbf{break}
    	\ENDIF
    \ENDWHILE
\ENDIF
\end{algorithmic}
\end{footnotesize}
\end{algorithm}

At each step, the supervisor collects the performance metrics $\eta_i$ from base stations and checks whether the performance of this period, $\eta^{(k)}$, has improved with respect to the previous period, $\eta^{(k-1)}$ (line 3). If this is the case, the system performance is raising and the supervisor increases TTI bounds $M_i$ as follows. The supervisor increases by $1$ unit the $M_i$ of the base station with the smallest $\eta_i$ whose $M_i$ is below $Z$ (lines 8-9). Once one $M_i$ value is increased, step $k$ of the algorithm terminates (line 10). 

If no $M_i$ can be increased, which means that all base stations are active in all TTIs, then no adjustment of the $M_i$ values is made as long as the system performance does not degrade. In case performance degrades, i.e., $\eta^{(k)}$ decreases, (line 13), the supervisor drastically reduces the $M_i$. Specifically, the supervisor looks at the base station $i$ with the largest $\eta_i$
whose $M_i$ is above $M^*_i$. It sets the new $M_i$ value of this station equal to the minimum between the half of the current $M_i$ value and the lower bound $M^*_i$ (lines 17-18). If $M_i = M^*_i$ for all $i$, no change is applied.

The rationale behind using AIMD to adjust the $M_i$ values is that, similar to what happens with Transmission Control Protocol (TCP), increasing the utilization of the system (i.e., increasing $M_i$ values) may lead to congestion (in our case, this corresponds to excessive interference), which causes user rate drops. In this case, a quick reaction is required by the supervisor to drive the system to a safe point of operation by properly adjusting TTI bounds $M_i$. Also similar to TCP, the additive increase of TTI bounds $M_i$ allows to gracefully approach the optimal utilization of the system. Furthermore, since the problem may admit more than one local maximum, using multiplicative decrease for the TTI bounds $M_i$ helps our heuristic to escape from a local maximum where the optimization function may be trapped in.

As a side comment, we point out that the proposed algorithm could accommodate different goals, such as, e.g., maximum throughput or proportional fairness, by simply replacing the function that gives the global system performance, $\eta$, by another function that reflects performance according to the objective pursued.

\subsection{DMS Control overhead}
\label{sec:mess_overhead}

In DMS, guaranteed and best-effort traffic procedures rely on the same exchanged information (ABSF patterns) with additional bits needed for BE to notify the base station performance $\eta_i$ to the supervisor.
We can identify two different {\it interfaces}: 
one between supervisor and base stations, namely $I_{C}$, and one between base stations, 
namely $I_{B}$. Counting the number of messages crossing each interface, we can estimate the required overhead carried by the backhauling/core network.

In the centralized solution, the supervisor requires message exchanges over $I_C$ only. 
In particular, per each pair ({\it user, base station}), it requires the transmission of an average channel quality indicator which can be encoded with $B$ bits as well as the (potential) incurred penalty $p_u$ with additional $B$ bits.
Then, the supervisor issues a scheduling pattern (a string of $W$ bits) to each base station. 

In the DMS mechanism, the supervisor requires the average user rate $\eta_i$ per base station over $I_C$ at the end of each game $\Omega$ for best-effort traffic only, consisting in a binary string of $B$ bits. Regarding the interface $I_B$ between different base stations, DMS needs a sequential exchange of ABSF patterns (strings of $W$ bits) during the interference coordination games $\Gamma$ and $\Omega$, until both games reach a convergence state or the convergence deadline expires.

We can therefore summarize the total load in terms of bits for each interface as reported in Table~\ref{t:overhead}.
\begin{table}
\centering
\caption{Overhead of centralized and DMS semi-distributed approaches}
\label{t:overhead}
	\footnotesize
  \begin{tabular}[c]{| c | c | c |}        
    \hline
    Interface & Centralized approach & DMS approach\\ \hline
    $I_{C}$ & $B\cdot \vert \mathcal{U} \vert\cdot \vert \mathcal{N} \vert + W\cdot \vert \mathcal{N} \vert$ & $2B \cdot \vert \mathcal{N} \vert$\\ \hline
    $I_{B}$ & $0$ & $W\cdot k\cdot \vert \mathcal{N} \vert$\\
    \hline
  \end{tabular}
\end{table}
In the table, $k$ is the number of rounds the interference coordination game plays before reaching the convergence, and $|\mathcal{U}| = \sum_i|\mathcal{U}_i|$ is the total number of users in the system. 
We can easily observe that the overhead of DMS is lower than that of the centralized 
mechanisms when the following inequality holds: 
\begin{equation}
\vert \mathcal{U} \vert > 1 + \frac{W}{B} (k-1) \cong \frac{W \vert \mathcal{N} \vert^2}{B}, 
\end{equation}
where we have considered that the number of rounds $k$ in the worst case is a function of $\vert \mathcal{N} \vert$ 
(i.e., at most $k=|\mathcal{N} \vert^2$ iterations are enough to converge, when convergence exists, as proven mathematically in~\cite{harks2013bottleneck} 
and empirically shown in Section~\ref{s:perf_inelastic} and Section~\ref{s:perf_elastic}) and both $W$ 
and $\vert \mathcal{N} \vert$ are (much) greater than 1. Therefore, our approach is convenient as soon as the number 
of users exceeds a threshold that depends on $W$, $B$, and $\vert \mathcal{N} \vert$ (i.e., the threshold is $O \left ( W \vert \mathcal{N} \vert^2 \right )$).
For example, in an (sub-)urban environment with $W \!=\! 70$, $\vert \mathcal{N} \vert=7$ and double precision floating point values, $B=64$, DMS results convenient with as few as 54 users or more, while in a dense-urban environment with $\vert \mathcal{N} \vert=28$, our approach exhibits a practical implementation starting with $\sim$900 users in the entire network. Those values are pretty low, revealing how \textit{our approach drastically reduces the signaling overhead for realistic cellular network sizes}.

Although we propose a high-level technology-independent analysis, we point our that DMS framework is compatible with the SDN paradigm~\cite{SDN4_crowd}. Moreover, interfaces $I_C$ and $I_B$ may be implemented using, e.g., the standard $X2$ interface~\cite{absf}. Since only ABSF patterns, penalty indicators and experienced rates need to be exchanged in addition to the measure of traffic demands received by each base station, the X2 interface could be adopted as Southbound interface in an SDN implementation of DMS with simple modifications.

\section{Performance Evaluation}
\label{s:results}
In this section, we use numerical simulations to show that our mechanism performs near optimally and boosts achievable rates in the  whole network, not just for topologically disadvantaged users. First, we present a simulation-based performance evaluation in which guaranteed traffic and best-effort traffic are independently served. Then, we show how DMS jointly handles both traffic types, exhibiting outstanding results. 

All simulations are carried out by means of MATLAB\textregistered\;with all parameters summarized in Table~\ref{t:sim_parameters}. Specifically, we tested two scenarios: a ``standard'' scenario with a $7$-base-station network deployed in an area of 300$\,$m$\times$500$\,$m, and a ``dense'' scenario with $28$ base stations in the same area. In both scenarios, base stations are deployed according to a hexagonal grid.\footnote{This assumption is rather common when carrying out a detailed simulation campaign with a significant number of base stations. However, it does not impact on the overall performance as our solution relies on an automated interaction between neighbouring cells, without any regards about the cell positions (such as hexagonal or PPP-based).}
The coverage of each base station is computed as a Voronoi region, assuming all base stations use the same transmission power $P\!=\!30$dBm. Although the scenario with $28$ base stations somehow includes the one with $7$---as ABSF could be used to blank all but $7$ base station in the denser scenario---we keep the two scenarios separated to analyze in a clean way the impact of both the number and the distance of interferers. Moreover, for the 7-base-station scenario, we are able to find and then compare optimal solutions with DMS solutions. Users are randomly dropped in each cell, according to a uniform random distribution. The average quality of the user channel is computed as function of the distance from the base station (according to the propagation model provided by 3GPP specifications, Table A.2.1.1-3 of TR.25.814 v7.1.0), and Rayleigh fading is considered. Based on user channel qualities, each simulated base station solves Problem~\texttt{BE-DISTR} or Problem~\texttt{GBR-DISTR} by means of a remote call to a commercial solver, i.e., IBM CPLEX OPL\textregistered. 
\begin{table} [!t]
\centering
\caption{List of Parameters for the LTE-A wireless scenarios used in the experiments [ITU UMi channel model]}
\label{t:sim_parameters}
	\footnotesize
  \begin{tabular}[c]{| c | c | c |}        
    \hline
    $\vert \mathcal{N} \vert$ & Number of Base Stations & $7$ or $28$\\ \hline
    $\vert \mathcal{U}_i \vert$ & Number of UEs per Base Station & $10-50$\\ \hline
    $W$ & ABSF Pattern Length & $70$ TTIs\\  \hline
    $BW$ & Spectrum Bandwidth & $20$ MHz\\  \hline
    $P$ & Transmitting Power & $30$ dBm\\  \hline
    $ISD$ & Inter-Site Distance & $80$ or $200$ m\\  \hline
    $N_0$ & Background Noise & $1.085 \times 10^{-14}$\\ \hline
  \end{tabular}
\end{table}

\subsection{Guaranteed traffic Management}
\label{s:perf_inelastic}
Here, we consider a variety of traffic demands and user distributions when DMS only handles guaranteed traffic requests. We show that DMS $(i)$ fully serves guaranteed user traffic demands, $(ii)$ minimizes resources used and make them available for best-effort traffic, and $(iii)$ performs close to the bound corresponding to the ideal centralized scheme, presented in Problem~\texttt{GBR} in Section~\ref{s:inelastic_problem_centr}. Finally, we show the rate of convergence of Distributed Inelastic Game $\Gamma$ proving that convergence is always achieved even in hostile cellular environments.

{\bf Time-resources performance.} 
To achieve traffic guarantees, we apply our DMS mechanism on the two scenarios described above, where all users in the network have subscribed a guaranteed bit-rate contract. For the sake of completeness we have simulated unbalanced scenarios where each base station exhibits a different GBR user demands. However, due to lack of space we have omitted this additional results, which only confirm the convenience of our approach showing near-optimal results without unveiling any critical issue.
Furthermore, we benchmark DMS against ``Centralized'', an omniscient centralized approach able to optimally solve Problem~\texttt{GBR} in the smallest possible number of TTIs, and ``Legacy'' corresponding to uncontrolled base stations using the same frequencies. 
Initially, the guaranteed traffic period is set equal to the ABSF pattern $W = 70$. Iteratively, DMS reduces the guaranteed traffic period while guaranteeing the required traffic. Presented results are averaged over $1000$ random  instances.

\begin{figure}[!t]
\centering
	\includegraphics[width=0.4\textwidth]{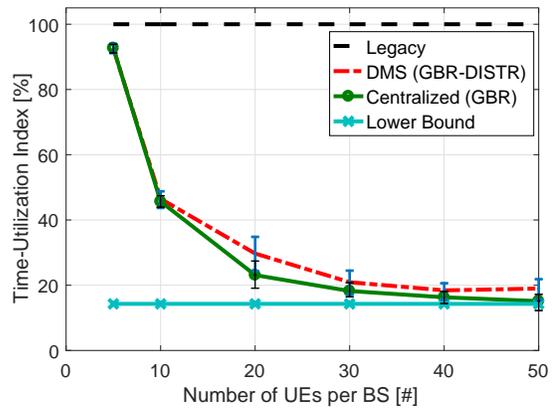}
	   \vspace{-2mm}
	\caption{\footnotesize Time-utilization index over different network density values for $4$ Mbps per GBR user and $7$ base stations. 
	}	
	\label{fig:reuse_factor}
\end{figure}

\begin{figure}[!t]
\vspace{-3mm}
\centering
	\includegraphics[width=0.4\textwidth]{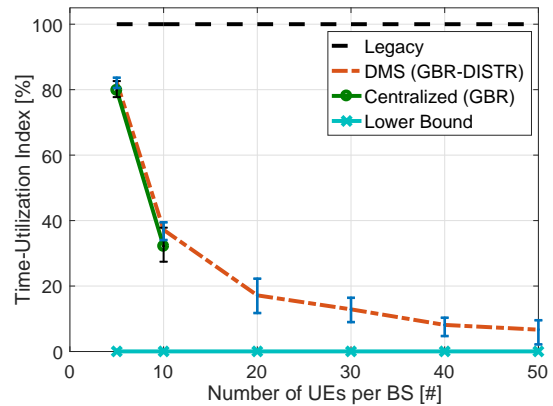}
	   \vspace{-2mm}
	\caption{\footnotesize Time-utilization index over different network density values for $3$ Mbps per GBR user and $28$ base stations. Optimal results are obtained only for few points due to the problem complexity.}	
	\label{fig:reuse_factor_dense}
\end{figure}

We study the amount of used resources under different network conditions in Fig.~\ref{fig:reuse_factor} and Fig.~\ref{fig:reuse_factor_dense}. 
They show 
the average number of TTIs used by each base station, 
normalized to the number of TTIs used in the legacy approach. The resulting value is a time-utilization index, which tells us how much each TTI is used by every base station in the network. The index ranges between $1/\vert \mathcal{N} \vert$ (marked as Lower Bound), when base stations use resources in a TDM-like way, and 1, like in the legacy scheme. The results are in line with our expectations. Considering a network with a very small user population, inter-cell interference is limited and the time-utilization index is close to its maximum ($92.72\%$ for Centralized and  $92.64\%$ for DMS). Conversely, for large user populations the overall interference grows and base stations are forced to transmit disjointly over the ABSF patterns ($45.65\%, 23.11\%, 18.23\%, 16.3\% \,\text{and}\, 15.15\%$ for Centralized and $46.49\%, 29.72\%, 20.93\%, 18.42\% \,\text{and}\, 18.99\%$ for DMS).
Fig.~\ref{fig:reuse_factor} shows that DMS is near-optimal, while the legacy approach wastes wireless-resources allowing base stations to be active all the time. Moreover, the figure confirms that DMS is not an excessively conservative approach, since it allows each BS to use slightly larger number of TTIs with respect to the operation of the optimal scheme. Fig.~\ref{fig:reuse_factor_dense} shows a similar performance in the dense scenario ($79.81\%, 32.14\%$ for Centralized and $82.31\%, 37.17\%$ for DMS). However, we could optimally solve Problem~\texttt{GBR} only for a small number of users per base station. The complexity of the centralized formulation is so high that we could not find a solution in a reasonable time for more than 10 users per base station.

In general, presented results confirm that DMS applied to guaranteed traffic substantially outperforms the legacy solution and exhibits near-optimal performance. The difference between our approach and the centralized scheme is typically very low (a few percent) and it is within $15\%$ in all cases.

\begin{table}[t!]
\caption{Study of convergence of game $\Gamma$ with different densities and and inter-site distances}
\vspace{-2mm}
\label{tab:convergence_study}
\scriptsize
\centering
\begin{tabular}{|c||c|c|c|c|}
\hline
\backslashbox{BS spacing\!\!\!}{\!\!\!\!\#Users/BS\!\!}
& $\textbf{10}$ & $\textbf{20}$ & $\textbf{30}$ & $\textbf{40}$ \\
\hline
\hline
$\textbf{100$\,$m}$ & $29.231\, (\pm 0.1)$ & $30.105$ & $\textbf{35.422}$ & $35.419$ \\
\hline
$\textbf{200$\,$m}$ & $7.111\, (\pm 0.01)$ & $10.708$ & $15.188$ & $17.495$ \\
\hline
$\textbf{500$\,$m}$ & $7.108\, (\pm 0.01)$ & $9.676$ & $15.951$ & $15.217$ \\
\hline
$\textbf{1000$\,$m}$ & $6.941\, (\pm 0.05)$ & $7.155$ & $7.244$ & $7.239$ \\
\hline
\end{tabular}
\end{table}

{\bf Convergence study of game $\Gamma$.} From the theoretical analysis presented in Section~\ref{s:distr_inel}, the Distribute Inelastic Game $\Gamma$ is guaranteed to converge by using SSBR, although it could very likely converge also by using a simple BR approach. Here, we experimentally evaluate the convergence properties of $\Gamma$ by simulating several network topologies wherein different network densities and inter-site distances are considered. In order to keep the computational cost affordable we only evaluate the 7-base-stations scenario assuming that our approach is applied to a cluster of cells for a dense scenario. Table~\ref{tab:convergence_study} summarizes the results obtained for the standard scenario in terms of $k$, the number of rounds the game needs to converge by using BR for at most $N^2$ rounds and switching to SSBR afterwards, if needed. Similar conclusions can be drawn in the very dense scenario.
The results are averaged over $1000$ simulations per each single case. It is important to note that $k$ decreases with the distance between BSs. This is due to the nature of inter-cell interference. The closer base stations are placed, the more the interference grows and the higher the number of rounds needed for $\Gamma$ to converge.
The dependence of $k$ on the user population is not so marked, although it is possible to state that $k$ roughly grows with the population size. 
Interestingly, even for dense scenarios, the  number of rounds $k$ that we have observed ranged from about $N$ to a maximum value which is lower than $N^2$ while the average reported in the table is of the same order of $N$.

Moreover, it is worth pointing out that only for $3$ cases out of $1000$ simulations the game $\Gamma$ did not reach the convergence by using the BR strategy, thus forcing players to use the SSBR strategy, as explained in Section~\ref{s:distr_inel}. SSBR is supposedly slow to reach convergence if used from round 1, however, it can readily lead to game convergence in a few rounds after BR has been played a few times. Specifically, in our simulation, 
about $N^2$ rounds with BR, followed by at most $N$ rounds with SSBR, were sufficient to reach convergence in all cases. 
These results confirm not only that convergence can be always achieved, but also that the BR strategy typically ensures the game convergence, with no need to force base stations to use the SSBR strategy since the beginning. In practice, we suggest to use the BR strategy during the first $N^2$ rounds of the game, and, if the game did not converged, switch to the SSBR strategy at round $N^2+1$. With the above, the entire game will converge in a number of rounds $O(N^2)$. 

\subsection{Best-effort traffic Management}
\label{s:perf_elastic}

Once guaranteed traffic is properly accommodated within the ABSF pattern $W$, the unused TTIs are fully assigned for serving best-effort traffic. In the next set of simulations, we show how DMS handles the best-effort traffic given a fixed number of available TTIs $Z<W$. 

First, we benchmark DMS against the optimal solution obtained by solving Problem \texttt{BE} by means of an ILP solver. 
Additionally, we compare DMS to a traditional ``Frequency Reuse 3'' scheme, in which the available band is split into three orthogonal sub-bands, and to the already mentioned ``Legacy'' case. 
Finally, for the sake of completeness, we compare DMS with two existing approaches based on power control schemes, showing how DMS can achieve high network performance at a bargain price of complexity.
In the first scheme, namely Utility-Based Power Control (UBPC)~\cite{shroff_ubpc}, base stations are allocated in all available TTIs by properly tuning the transmitted power to reduce interference. The algorithm suggested in~\cite{shroff_ubpc} maximizes the user net utility by ensuring that the signal-to-noise-ratio of each transmission is greater than a minimum threshold $\gamma_i$ (in our simulations we assume $\gamma_i$ as the minimum MCS with nonzero rate). 
While UBPC provides a rigorous centralized solution for the power allocation problem at the expense of a huge amount of information exchanged, a second power control scheme recently developed, namely REFerence based Interference Management (REFIM)~\cite{refim}, proposes a low-complex distributed scheme by exploiting the notion of reference user (e.g., the user with the worst channel condition belonging to the surrounding cells). 
Although this abstraction leads to a drastic reduction of the control signal overhead and results in a practical implementation of the power control solution, it exhibits a conservative behaviour.

{\bf Utility and fairness performance.} We start evaluating the system utility $\widehat{\eta}$, which, according to the formulation of Problem \texttt{BE}, is the sum of minimum user rates experienced at each base station. Fig.~\ref{fig:dynamics} shows $\widehat{\eta}$ for the above-described $28$-cell scenario when different schemes are applied. 

\begin{figure} [!t]
\centering
	\includegraphics[width=0.4\textwidth]{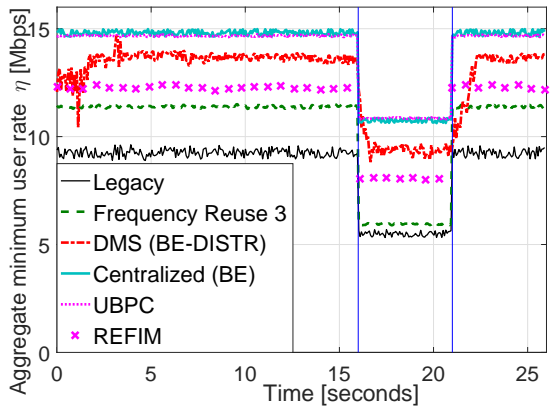}
		\caption{\footnotesize Dynamic behaviour of DMS for best-effort traffic applied to a dense scenario. In the first part, the scenario includes $\vert \mathcal{N} \vert = 28$ base stations, $|U_i| = 20$ users and $Z = 140$ TTIs. A network change increases the number of users up to $|U_i| = 30$ users and then back again to $|U_i| = 20$ users. 
}
	\label{fig:dynamics}
\end{figure}

\begin{figure} [!t]
\centering
	\includegraphics[width=0.4\textwidth]{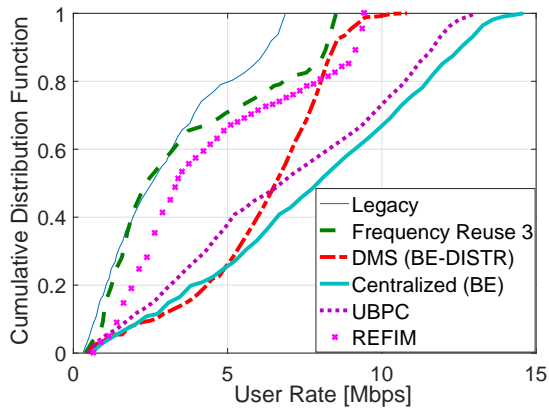}
	\caption{\footnotesize CDF of average user rates with 28 base stations and 20 BE users per base station. The time horizon is set to $Z = 140$ TTIs.}
	\label{fig:user_cdf}
\end{figure}
 
Due to the adaptive nature of our algorithm, DMS shows a dynamic behavior. Specifically, it takes a few seconds for DMS to reach its stable operating point, after which it follows quite fast the evolution of channel and traffic conditions. In particular, at time $t \! = \! 16$ s, the number of users in the network increases by $50\%$, but it takes only a fraction of a second for DMS to adapt. The same happens at $t \! = \! 21$ s, when the number of users returns to the initial value. In general, DMS largely outperforms the Legacy scheme and achieves significant gain over Frequency Reuse 3, being much closer to the optimal performance. 
Notably, after the initial adaptation period, the utility achieved by DMS is more than $90\%$ of the one achieved with the optimal solution for Problem~\texttt{BE}.
REFIM and UBPC results show the real potentials of power control schemes. UBPC can very closely match the performance of the optimal solution without power control, although it requires higher complexity in terms both of execution and device hardware. REFIM, despite being a conservative low-complexity scheme, shows a better performance than Frequency Reuse 3. Our DMS approach perfectly lies in between an unviable efficient power control scheme and a practically doable distributed power control solution. Therefore, considering the complexity of power control schemes, DMS can provide a practical and advantageous trade-off between performance and complexity.
 
Besides utility $\eta$, we want to evaluate the fairness achieved by the different schemes. To this aim, Fig.~\ref{fig:user_cdf} presents the CDF of achieved user rates (averaged over the time horizon $W$). The figure clearly shows that the optimal centralized solution and UBPC can achieve higher user rates, while Legacy, Frequency Reuse 3, and REFIM cannot fully support this dense network scenario. DMS provides intermediate user rates, however, it has the strong advantage of presenting them in a compact interval of possible values, which is symptom of fairness. Moreover, although the dense scenario causes high congestion, DMS can provide a boost to the critical rates experienced by the users. For instance, with  90\% probability, DMS guarantees 2.1 Mbps per user, UBPC guarantees 1.8 Mbps per user, REFIM guarantees 1.3 Mbps per user, while Frequency Reuse 3 only guarantees 0.8 Mbps.

{\bf Convergence study of game $\Omega$.} A key feature of DMS is its ability to quickly adapt to network changes. Such a feature relies on a fast ABSF pattern computation, which follows the rule of the Interference Coordination Game $\Omega$, and on convergence guarantees.
\begin{figure} [!t]
\centering
	\includegraphics[width=0.4\textwidth]{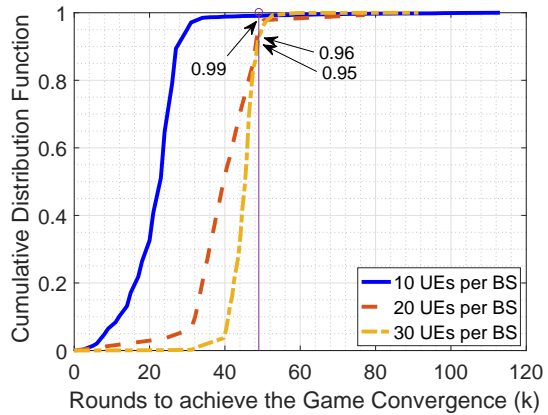}
	\caption{\footnotesize CDF of number of rounds needed for game convergence with 28 base stations and different user populations.}
	\label{fig:game_cdf}
  \end{figure}
Interestingly, Fig.~\ref{fig:game_cdf} illustrates the CDF of the number of rounds needed to converge for different user populations. Please note that in the dense scenario DMS approach is distributively applied to a cluster of $7$ base stations in order to keep the computational burden feasible. Therefore, as shown in the figure, the majority of the games $\Omega$ converge much before $\vert \mathcal{N} \vert^2$ rounds (vertical line in the figure), and very few cases do not converge at all. Simulations also show that the rate improvement beyond the $\vert \mathcal{N} \vert^2$-th round is very limited. Therefore,  we can conclude that reasonably high utilities can be achieved by stopping the game after a number of rounds comprised between $\vert \mathcal{N} \vert$ and $\vert \mathcal{N} \vert^2$.  

Overall, when best-effort traffic is involved, our results show that DMS not only achieves near-optimal results according to the definition of utility given in the formulation of Problem~\texttt{BE}, but also achieves high levels of fairness, and significantly boosts average rates in the entire cellular network. 

\subsection{Multi-traffic Management}
\label{sec:jointly_dms}

We have carried out an empirical evaluation to show how DMS can minimize the time needed to fully serve guaranteed traffic and simultaneously optimize the scheduling of best-effort traffic demands in the remaining part of the ABSF pattern. Fig.~\ref{fig:dms_merged} and Fig.~\ref{fig:dms_dense_merged} show two representative examples of obtained results. We consider the DMS mechanism for guaranteed traffic always applied, while alternative approaches are evaluated for best-effort traffic. Indeed, since the mechanism for guaranteed traffic has been exhaustively evaluated in Section~\ref{s:perf_inelastic}, the goal of these experiments is to show that GBR and BE algorithms can dynamically interact to optimize the exploitation of available resources.

In Fig.~\ref{fig:dms_merged}, we consider a scenario with $7$ base stations and a total of $70$ GBR users each requiring $4 Mbps$, $140$ BE users, and a time horizon of $W=70$ TTIs. At $t=119 s$, the scenario changes by varying the GBR users' demands from $4 Mbps$ to $5Mbps$. In Fig.~\ref{fig:dms_dense_merged}, instead, the scenario is much denser, having $28$ base stations, $560$ GBR users with a $2 Mbps$ demand, and $140$ BE users. The time horizon is $W=140$ and users' demands change from $2 Mbps$ to $2.5 Mbps$. The different parameter is required to make the traffic demand feasible for the available network capacity in such a dense condition. Both figures report four curves: ``GBR (DMS)'' refers to the total throughput of guaranteed-traffic users computed according to DMS guaranteed traffic mechanism, ``BE (DMS)'' refers to the total throughput of best-effort users considering both the current DMS allocation for guaranteed traffic and the current iteration of BE mechanism, ``BE (Legacy)'' considers the current DMS allocation for guaranteed traffic and the Legacy approach for the best-effort traffic, ``BE (Centralized)'' still considers the current DMS allocation for guaranteed traffic, while providing the centralized optimum of Problem~\texttt{BE} for the best effort traffic.

The figures show that: $(i)$ guaranteed traffic is always served providing exactly the requested throughput and $(ii)$ DMS finds the optimal number of TTIs to fully accommodate the guaranteed traffic after few seconds, thus it dynamically schedules the best-effort traffic. DMS dramatically outperforms the legacy approach by unveiling near-optimal results when the game $\Omega$ convergence is reached.  Once the traffic demand of guaranteed traffic users increases at $t=119s$, DMS reserves the full ABSF pattern for the GBR leaving no space to best-effort demands. Then, the optimization restarts, showing again outstanding results when compared to a legacy solution and closely following the optimal results achieved by solving the centralized Problem~\texttt{BE}. 

We can finally state that the advantage brought by DMS is threefold: ($i$) near-optimal results are easily achieved after few seconds without requiring high-complexity solutions, ($ii$) the overall system capacity is drastically increased due to a smart usage of system resources and to a strong limitation of the inter-cell interference, and ($iii$) it can easily adapt to scenario changes supporting guaranteed traffic and smartly managing best-effort requests.

\begin{figure} [!t]
\centering
	\includegraphics[width=0.4\textwidth]{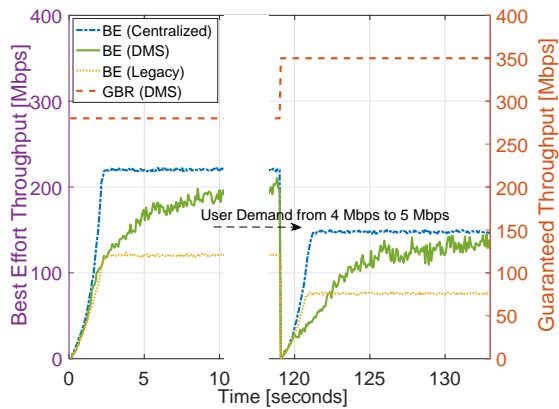}
	\caption{\footnotesize System throughput of DMS for both guaranteed and best effort traffic in the standard network scenario.
	}
	\label{fig:dms_merged}
\end{figure}

\begin{figure} [!t]
\centering
	\includegraphics[width=0.4\textwidth]{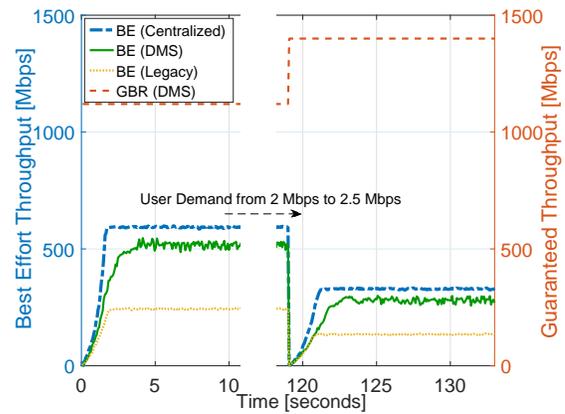}
	\caption{\footnotesize System throughput of DMS for both guaranteed and best effort traffic in the dense network scenario.
	}
	\label{fig:dms_dense_merged}
\end{figure}

\section{Related work}
\label{s:related}

The ABSF technique is becoming popular because it is suitable for eICIC in LTE, it requires minimal changes in the operation of base stations and offers flexible tools to trade-off between performance improvement and implementation 
complexity \cite{wang2011time, ghosh2012heterogeneous}. However, designing a mechanism to drive ABSF decisions has turned out to be challenging and multifaceted. For instance, the authors of \cite{cierny2013on} have studied quantitative approaches aiming to determine the best {\it density} of blanked subframes as a function of the traffic distribution.
Other studies focus on heterogeneous scenarios where a macro base station and several small base stations have to coordinate their activities using ABSF patterns \cite{Lei2012resource,kamel2012performance,kamel2013absf}. Also, ICIC approaches in time domain have been investigated in~\cite{enhanced_icic_hetnets}.
Other proposals include access selection in the loop and introduce the concept of Cell Selection Bias \cite{deb2014algorithms}, which improves network spectral efficiency \cite{singh2014joint,yinghao2013joint}. 
However, existing ABSF solutions either require a central entity to gather per-user CSI 
or need additional and continuously updated information on, e.g., topology and propagation environment, which goes well beyond current base station's features and capabilities. As a result, existing ABSF solutions are not scalable and do not adapt quickly when network conditions change. 

Some other solutions for resource management behave similarly to ABSF. For instance, a recent proposal for OFDMA femtocells has been presented in~\cite{yoon2014self}. Although the authors do not explicitly use the ABSF paradigm, their work is based on detecting the best region of the time-frequency space where base stations can transmit, like in ABSF. However, they propose a {\it probe-and-adapt} algorithm to decide whether to use or blank resources. Moreover, they do not require coordination between base stations, which would yield performance limitations. Similarly, the authors 
of~\cite{son2011utility} propose the concept of {\it reuse patterns} for base station activities, which clearly mimics ABSF operations. However, their work focuses on finding the best temporal duration of reuse patterns (in order to maximize the total user throughput) but it does not explain how to generate reuse patterns. Moreover, differently from DMS, the proposal presented in~\cite{son2011utility} does not take into account fairness.
 
The authors of~\cite{ellenbeck2008decentralized} present a game theoretical approach to ICIC. Their approach addresses the coordination among base stations over a set of finite resources as a non-cooperative game. However, they only target the minimization of  the perceived interference, and do not take into account user scheduling. 
 
None of the above proposals embodies the set of features that characterize our approach and can be summarized as follows: $(i)$ the distributed ABSF interference coordination problem has been formalized and its convergence investigated, $(ii)$ the proposed mechanism is semi-distributed and the complexity of the centralized coordination 
is abated and split among base stations requiring little signaling exchange with the central coordinator,  $(iii)$ the proposed mechanism is adaptive and can adjust its parameters according to traffic dynamics, $(iv)$ despite the simplicity of the proposed mechanism, our results show remarkable near-optimal performance figures and, $(v)$ at best of our knowledge, there is no literature on the use of ABSF techniques for properly serving inelastic traffic. 

\section{Conclusions}
\label{s:conclusions}

In this paper, we have presented the design of DMS, a practical (distributed and lighweight) approach to optimize inter-cell interference coordination for both guaranteed traffic and best-effort traffic. To design this approach, we have formulated two optimization problems, one for each traffic type, relying on game theory notions. We have then proposed distributed algorithms to solve these optimization problems, and have further conducted analysis to prove the convergence and stability of these algorithms. As a result, with our approach base stations make scheduling decisions for serving guaranteed traffic by using as few TTIs as possible, leaving the room for best-effort traffic, which is efficiently served. 

Due to the simplicity of our approach and its limited control overhead, this is to the best of our knowledge a first attempt towards an efficient, scalable and adaptive implementation of ABSF that simultaneously addresses multiple traffic types and provides a viable solution to be deployed in real networks. 
Our numerical results show that DMS achieves near-optimal results with respect to a centralized omniscient network scheduler, and achieves performance levels similar to schemes relying on complex power control functionality. 


\bibliographystyle{IEEEtran}
\bibliography{biblioinfocom2014,bibliography}

\begin{thebibliography}{10}
\providecommand{\url}[1]{#1}
\csname url@samestyle\endcsname
\providecommand{\newblock}{\relax}
\providecommand{\bibinfo}[2]{#2}
\providecommand{\BIBentrySTDinterwordspacing}{\spaceskip=0pt\relax}
\providecommand{\BIBentryALTinterwordstretchfactor}{4}
\providecommand{\BIBentryALTinterwordspacing}{\spaceskip=\fontdimen2\font plus
\BIBentryALTinterwordstretchfactor\fontdimen3\font minus
  \fontdimen4\font\relax}
\providecommand{\BIBforeignlanguage}[2]{{%
\expandafter\ifx\csname l@#1\endcsname\relax
\typeout{** WARNING: IEEEtran.bst: No hyphenation pattern has been}%
\typeout{** loaded for the language `#1'. Using the pattern for}%
\typeout{** the default language instead.}%
\else
\language=\csname l@#1\endcsname
\fi
#2}}
\providecommand{\BIBdecl}{\relax}
\BIBdecl

\bibitem{cisco}
Cisco, ``Cisco visual networking index: Global mobile data traffic forecast
  update, 2013-2018,'' \emph{White paper}, February 2014.

\bibitem{5G-MiWEBA}
K.~Sakaguchi, G.~K. Tran, H.~Shimodaira, S.~Nanba, T.~Sakurai, K.~Takinami,
  I.~Siaud, E.~C. Strinati, A.~Capone, I.~Karls, R.~Arefi, and T.~Haustein,
  ``Millimeter-wave evolution for 5g cellular networks,'' \emph{IEICE Trans. on
  Communications}, vol. E98-B, no.~3, pp. 388--402, 2015.

\bibitem{LTEadv}
{Third Generation Partnership Project (3GPP)}, ``{Requirements for further
  advancements for Evolved Universal Terrestrial Radio Access (E-UTRA)
  (LTE-Advanced)},'' {3GPP TS 36.913 v 13.0.0}, January 2016.

\bibitem{absf}
------, ``{Evolved Universal Terrestrial Radio Access Network (E-UTRAN); X2
  application protocol (X2AP)},'' {3GPP TS 36.423 v. 14.0.0}, September 2016.

\bibitem{GMO_2016}
O.~Grøndalen, K.~Mahmood, and O.~N. Østerbø, ``On attachment optimization
  and muting pattern selection in eicic,'' in \emph{2016 IEEE 27th Annual
  International Symposium on Personal, Indoor, and Mobile Radio Communications
  (PIMRC)}, Sept 2016, pp. 1--6.

\bibitem{SciancaleporeMBZCP14}
V.~Sciancalepore, V.~Mancuso, A.~Banchs, S.~Zaks, and A.~Capone, ``Enhanced
  content update dissemination through {D2D} in {5G} cellular networks,''
  \emph{IEEE Transactions on Wireless Communications}, vol.~15, pp. 7517--7530,
  Nov 2016.

\bibitem{SMB13}
V.~Sciancalepore, V.~Mancuso, and A.~Banchs, ``{BASICS}: Scheduling base
  stations to mitigate interferences in cellular networks,'' in \emph{IEEE
  WoWMoM 2013}.

\bibitem{xueying2011cell}
X.~Hou, E.~Bjornson, C.~Yang, and M.~Bengtsson, ``Cell-grouping based
  distributed beamforming and scheduling for multi-cell cooperative
  transmission,'' in \emph{IEEE PIMRC 2011}.

\bibitem{irmer2011coordinated}
R.~Irmer, H.~Droste, P.~Marsch, M.~Grieger, G.~Fettweis, S.~Brueck, H.~Mayer,
  L.~Thiele, and V.~Jungnickel, ``Coordinated multipoint: Concepts,
  performance, and field trial results,'' \emph{IEEE Communications Magazine},
  vol.~49, no.~2, pp. 102--111, 2011.

\bibitem{ic_ofdma}
M.~C. Necker, ``{Interference Coordination in Cellular OFDMA Networks},''
  \emph{IEEE Network}, vol.~22, no.~6, p.~12, December 2008.

\bibitem{singh2014joint}
S.~Singh and J.~Andrews, ``Joint resource partitioning and offloading in
  heterogeneous cellular networks,'' \emph{IEEE Transactions on Wireless
  Communications}, vol.~13, no.~2, pp. 888--901, February 2014.

\bibitem{yinghao2013joint}
Y.~Jin and L.~Qiu, ``Joint user association and interference coordination in
  heterogeneous cellular networks,'' \emph{IEEE Communications Letters},
  vol.~17, no.~12, pp. 2296--2299, December 2013.

\bibitem{son_commag}
C.~Prehofer and C.~Bettstetter, ``Self-organization in communication networks:
  principles and design paradigms,'' \emph{Communications Magazine, IEEE},
  vol.~43, no.~7, pp. 78--85, July 2005.

\bibitem{shroff_ubpc}
M.~Xiao, N.~Shroff, and E.~K.~P. Chong, ``A utility-based power-control scheme
  in wireless cellular systems,'' \emph{Networking, IEEE/ACM Transactions on},
  vol.~11, no.~2, pp. 210--221, Apr 2003.

\bibitem{refim}
K.~Son, S.~Lee, Y.~Yi, and S.~Chong, ``{REFIM}: A practical interference
  management in heterogeneous wireless access networks,'' \emph{Journal on
  Selected Areas in Communications}, vol.~29, no.~6, pp. 1260--1272, 2011.

\bibitem{LTEadv36213}
{Third Generation Partnership Project (3GPP)}, ``{Physical layer procedures
  (Release 14) for Evolved Universal Terrestrial Radio Access (E-UTRA)},''
  {3GPP TS 36.213 v. 14.0.0}, September 2016.

\bibitem{goussevskaia:capacity}
O.~Goussevskaia, R.~Wattenhofer, M.~M. Halldorsson, and E.~Welzl, ``Capacity of
  arbitrary wireless networks,'' in \emph{IEEE INFOCOM 2009}.

\bibitem{Bertsekas}
D.~Bertsekas and R.~Gallager, \emph{Data Networks (II Ed.)}.\hskip 1em plus
  0.5em minus 0.4em\relax Prentice-Hall.

\bibitem{harks2013bottleneck}
T.~Harks, M.~Hoefer, M.~Klimm, and A.~Skopalik, ``Computing pure nash and
  strong equilibria in bottleneck congestion games,'' \emph{Springer
  Mathematical Programming}, vol. 141, no. 1-2, pp. 193--215, 2013.

\bibitem{secon15_sfmcb}
V.~Sciancalepore, I.~Filippini, V.~Mancuso, A.~Capone, and A.~Banchs, ``A
  semi-distributed mechanism for inter-cell interference coordination
  exploiting the {ABSF} paradigm,'' in \emph{IEEE SECON 2015}.

\bibitem{LMR_EWSDN12}
L.~E. Li, Z.~M. Mao, and J.~Rexford, ``Toward software-defined cellular
  networks,'' in \emph{2012 European Workshop on Software Defined Networking},
  Oct 2012, pp. 7--12.

\bibitem{5gnorma}
{5G Novel Radio Multiservice adaptive network Architecture (5G NORMA)},
  ``{Deliverable D3.2 - 5G NORMA network architecture – Intermediate
  report},''
  \url{https://5gnorma.5g-ppp.eu/wp-content/uploads/2017/03/5g_norma_d3-2.pdf/},
  Jan 2017.

\bibitem{aimd_perf}
C.~Liu and E.~Modiano, ``On the performance of additive increase multiplicative
  decrease ({AIMD}) protocols in hybrid space-terrestrial networks,''
  \emph{Comput. Netw.}, vol.~47, no.~5, pp. 661--678, Apr. 2005.

\bibitem{SDN4_crowd}
H.~Ali-Ahmad, C.~Cicconetti, A.~De~La~Oliva, M.~Draxler, R.~Gupta, V.~Mancuso,
  L.~Roullet, and V.~Sciancalepore, ``{CROWD: An SDN Approach for DenseNets},''
  in \emph{Software Defined Networks (EWSDN), 2013 Second European Workshop
  on}, Oct 2013, pp. 25--31.

\bibitem{wang2011time}
Y.~Wang and K.~Pedersen, ``Time and power domain interference management for
  {LTE} networks with macro-cells and {HeNBs},'' in \emph{IEEE VTC 2011}.

\bibitem{ghosh2012heterogeneous}
A.~Ghosh, N.~Mangalvedhe, R.~Ratasuk, B.~Mondal, M.~Cudak, E.~Visotsky, T.~A.
  Thomas, J.~G. Andrews, P.~Xia, H.~S. Jo, H.~S. Dhillon, and T.~D. Novlan,
  ``Heterogeneous cellular networks: From theory to practice,'' \emph{IEEE
  Communications Magazine}, no.~6, pp. 54--64, 2012.

\bibitem{cierny2013on}
M.~Cierny, H.~Wang, R.~Wichman, Z.~Ding, and C.~Wijting, ``On number of almost
  blank subframes in heterogeneous cellular networks,'' \emph{IEEE Trans. on
  Wireless Comm.}, vol.~12, no.~10, pp. 5061--5073, Oct. 2013.

\bibitem{Lei2012resource}
L.~Jiang and M.~Lei, ``Resource allocation for {eICIC} scheme in heterogeneous
  networks,'' in \emph{IEEE PIMRC 2012}.

\bibitem{kamel2012performance}
M.~Kamel and K.~M. Elsayed, ``Performance evaluation of a coordinated
  time-domain {eICIC} framework based on {ABSF} in heterogeneous {LTE-Advanced}
  networks,'' in \emph{IEEE GLOBECOM 2012}.

\bibitem{kamel2013absf}
M.~Kamel and K.~Elsayed, ``{ABSF} offsetting and optimal resource partitioning
  for {eICIC} in {LTE-Advanced}: Proposal and analysis using a nash bargaining
  approach,'' in \emph{IEEE ICC 2013}, pp. 6240--6244.

\bibitem{enhanced_icic_hetnets}
D.~Lopez-Perez, I.~Guvenc, G.~de~la Roche, M.~Kountouris, T.~Q.~S. Quek, and
  J.~Zhang, ``Enhanced intercell interference coordination challenges in
  heterogeneous networks,'' \emph{IEEE Wireless Communications}, vol.~18,
  no.~3, pp. 22--30, June 2011.

\bibitem{deb2014algorithms}
S.~Deb, P.~Monogioudis, J.~Miernik, and J.~Seymour, ``Algorithms for enhanced
  inter-cell interference coordination ({eICIC}) in {LTE} hetnets,''
  \emph{IEEE/ACM Trans. on Networking}, vol.~22, no.~1, pp. 137--150, Feb 2014.

\bibitem{yoon2014self}
J.~Yoon, M.~Y. Arslan, K.~Sundaresan, S.~V. Krishnamurthy, and S.~Banerjee,
  ``Self-organizing resource management framework in {OFDMA} femtocells,''
  \emph{IEEE Trans. on Mobile Computing}, vol.~14, no.~4, pp. 843--857, 2015.

\bibitem{son2011utility}
K.~Son, Y.~Yi, and S.~Chong, ``Utility-optimal multi-pattern reuse in
  multi-cell networks,'' \emph{IEEE Trans. on Wireless Communications},
  vol.~10, no.~1, pp. 142--153, 2011.

\bibitem{ellenbeck2008decentralized}
J.~Ellenbeck, C.~Hartmann, and L.~Berlemann, ``Decentralized inter-cell
  interference coordination by autonomous spectral reuse decisions,'' in
  \emph{European Wireless}, 2008, pp. 1--7.

\end{thebibliography}

\clearpage 

\appendix
\begin{figure}[!t]
\large
\bf
\centering
Supplementary downloadable material 
\end{figure}


\label{s:appendix}
\begin{replemma}{Theo:convergence}
Given that the players' actions belong to whatever action profile $\sigma$, after a finite number of single-step best responses (SSBR), all players' actions will belong to a saturation action profile $\overline{\sigma}$.
\end{replemma}

\begin{IEEEproof}
We prove the lemma by contradiction.
Let us define two sets of players that represent the state of the game at a given point before reaching a saturation action profile $\overline{\sigma}$. 
Specifically, set $\mathcal{P}$ includes all players that have penalties, and set $\mathcal{\bar{P}}$ includes the remaining players, which incur zero penalty. 
Assume now that all players switch to Single-Step Best Response (SSBR) at a given point in time, and consider the composition of $\mathcal{P}$ and $\mathcal{\bar{P}}$ at that point. 

If the lemma were incorrect, the game could evolve from this state and at least one player could not reach a saturation action profile $\overline{\sigma}$.
This possibility implies that sets $\mathcal{P}$ and $\mathcal{\bar{P}}$ are always nonempty at the end of each round of the game. If not, the evolution of the game would lead to have either  - case $c.ii$ - all players using all the TTIs (all players being in $\mathcal{P}$ would cause a progressive increment in the use of TTIs until all TTIs are used by all players), or - case $c.i$ - all players incurring no penalty (all players in $\mathcal{\bar{P}}$ would reduce the use of TTIs as much as possible, without incurring in penalties).
Note that a third case is possible - case $c.iii$ -, in which a subset of users is either in case $c.ii$ or in $c.i$, while the remaining ones, although incurring in penalties, have no incentive to change their strategy. Indeed, further adding slots does not produce an objective function improvement, as null contributions are added, whereas limiting the number of occupied slots surely does not reduce the penalty. In all cases, we would have a saturation action profile for all players.

Moreover, there must exist a continuous flow of players moving between $\mathcal{P}$ and $\mathcal{\bar{P}}$ while the game evolves. In fact, should the flow stop after a finite number of actions, all players in $\mathcal{P}$ would increase the number of TTIs used until they reach the maximum (because they have penalties to pay) or no further addition can be beneficial because of the null contribution, and all players in  $\mathcal{\bar{P}}$ would decrease the number of TTIs used without incurring any penalty. However, by definition, this would be a saturation action profile $\overline{\sigma}$. 

Therefore, to admit that SSBR does not lead to a saturation action profile, we have to admit that players action continuously between $\mathcal{P}$ and $\mathcal{\bar{P}}$. Moreover, since the two sets have to remain nonempty at the end of any round of the game, 
the flows of players from  $\mathcal{P}$ to $\mathcal{\bar{P}}$ and vice versa have to be balanced. With two possible states, 
this also implies that the probability to be in $\mathcal{P}$ is the same as the probability to be in $\mathcal{\bar{P}}$; therefore, the average sojourn times of a player in $\mathcal{P}$ and $\mathcal{\bar{P}}$ are the same.
Note that in the case every user is in the condition $c.iii$, we have reached an equilibrium. Indeed, no user can improve its objective function by unilaterally deviating from its strategy. Vice versa, users not satisfying condition $c.iii$ behave as follows.

In particular, let us consider a player $p$ that keeps moving between $\mathcal{P}$ and $\mathcal{\bar{P}}$, and let us call $d$ the average time spent in each of the two states. Let us consider a {\it cycle} of player $p$, from her passage to $\mathcal{P}$ to her return to $\mathcal{\bar{P}}$. 
In the transition $\mathcal{\bar{P}} \rightarrow \mathcal{P}$, $p$ will increase the TTI utilization by 1 unit, to try to come back to $\mathcal{\bar{P}}$ immediately. Then she will spend $d-1$ rounds in 
$\mathcal{P}$, during which she increments by $d-1$ units her TTI utilization. Afterwards, $p$ goes back to $\mathcal{\bar{P}}$, which
can occur with an increase of one TTI or with no changes (because of other player's changes of action). Eventually, player $p$ will spend $d-1$ rounds in $\mathcal{\bar{P}}$, during which she will decrement the TTI utilization by {\it at most} $d-1$ units. The resulting balance is a net increase in the number of used TTIs. Therefore, all players moving continuously between $\mathcal{P}$ and $\mathcal{\bar{P}}$ should eventually end up using all the available TTIs and have no way to further change their action profiles. Hence, the flow between $\mathcal{P}$ and $\mathcal{\bar{P}}$ would stop. This would lead again to a saturation action profile $\overline{\sigma}$. 
\end{IEEEproof}

\begin{replemma}{Theo:sat_point}
At a certain point in time, given that the actions played by any player in the system belong to a saturation action profile $\overline{\sigma}$ if each player chooses a single-step best response (SSBR), the game will converge to a Nash equilibrium.
\end{replemma}
\begin{IEEEproof}
We can derive from Problem~\texttt{GBR-DISTR} the cost function $f_i$ related to a player action $S_i$ and the actions taken by the other players ($S_{-i}$) as follows
\begin{multline}
\label{cost_function}
f_i(S_i,S_{-i}) = |S_i| + \alpha \cdot \sum_{u \in U_i} \rho_u (S_i,S_{-i}), \quad \forall S_i \in \mathbb{S}_i; \\
\rho_u(S_i,S_{-i}) = \max \left(D_u - \sum_{(u,t) \in S_i} c_{u,t} (S_{-i}),0 \right),
\end{multline}
where $\rho_{S_i,S_{-i}}$ is the penalty that user $u$ has to pay in order to satisfy its user traffic demand $D_u$. 
It is clear that each player chooses her single-step best response in order to minimize $f_i(S_i,S_{-i})$. Due to the saturation action profile, if player $i$ presents at step $k-1$ a zero penalty, all users' traffic demands $D_u,~\forall u \in U_i$, are satisfied with the current player's action $S_i^{(k-1)}$. Noticing that the cost function is not decreasing with users' action cardinality, in the case of saturation with zero penalty, the only relevant term in the cost function is the cardinality of the current action $|S_i|$ (i.e., the number of TTIs used for scheduling the users). Hence, at next step $k$, player $i$ will choose a action $S_i^{(k)}$ such that $|S_i^{(k)}| \! \le \! |S_i^{(k-1)}|$ due to the single-step best response, which leads to
\begin{equation}
\label{theo1_cost}
f(S_i^{(k)}, S_{-i}) \leq f(S_i^{(k-1)}, S_{-i}),
\end{equation}
\noindent where $f(\cdot)$ is the cost function defined in~\eqref{cost_function}. 
Given that the player change $S_i^{(k-1)}$ will not increase the inter-cell interference towards the other cells, it may benefit the other players choices. Hence, the penalty value in the cost function will not ever be increased by the other players, and the updated action profile $\sigma^{(k)}$ is still a saturation action profile at step $k$.
Therefore, we deduce that the inequality~\eqref{theo1_cost} will be satisfied for all players' actionss in the system, at any step $k$. 

Since we assume a non-decreasing cost function, each player will get the minimum of the cost function in a finite number of steps. Upon all players choose the particular action returning the minimum of the cost function, they have reached a Nash equilibrium. Furthermore, if players' actions take all available $W$ TTIs with a non-zero penalty,
the players have already reached a Nash equilibrium.
Since they cannot further increase the number of involved TTIs, no further action will improve their cost function. 
The same holds for the users with a non-zero penalty and no objective function improvement if further TTIs were added.
\end{IEEEproof}

\end{document}